\documentclass[runningheads,a4paper]{llncs}

\usepackage{amssymb}
\usepackage{graphicx}
\newcommand{\negA}{\vspace{-0.05in}}
\newcommand{\negB}{\vspace{-0.1in}}
\newcommand{\negC}{\vspace{-0.18in}}

\newcommand{\mysection}[1]{\negC\section{#1}\negA}
\newcommand{\mysubsection}[1]{\negB\subsection{#1}\negA}

\newcommand{\myparagraph}[1]{\par\smallskip\par\noindent{\bf{}#1:~}}
\newtheorem{obs}{Observation}
\newtheorem{cla}{Claim}
\newtheorem{mylemma}[theorem]{Lemma}

\usepackage{algorithm}
\usepackage{algorithmic}
\newcommand{\alg}[1]{\mbox{\sf #1}}  

\newcommand{\comment}[1]{}

\usepackage{amsopn}

\usepackage{url}
\urldef{\mails}\path|{hadas, meizeh}@cs.technion.ac.il|

\begin{document}

\mainmatter

\title{Representative Families: A Unified Tradeoff-Based Approach}

\author{Hadas Shachnai \and Meirav Zehavi}

\institute{Department of Computer Science, Technion, Haifa 32000, Israel\\
\mails}

\maketitle

\vspace{-0.93em}

\begin{abstract}
Let$\ \!M\!=\!(\!E\!,{\cal I})\ \!$be a matroid, and let ${\cal S}$ be a family of subsets\\of size $p$ of $E$. A subfamily $\widehat{\cal S}\!\subseteq\!{\cal S}$ represents ${\cal S}$ if for every~pair~of~sets $X\!\in\!{\cal S}$ and $Y\!\subseteq\!E\!\setminus\! X$ such that $X\!\cup\!Y\!\in\!{\cal I}$, there is~a~set~$\widehat{X}\!\in\!\widehat{\cal S}$ disjoint from $Y$ such that $\widehat{X}\!\cup\!Y\!\in\!{\cal I}$.
Fomin et al. ({\em Proc. ACM-SIAM Symposium on Discrete Algorithms, 2014}) introduced a powerful technique for fast computation of representative families for uniform matroids. In this paper, we show that this technique leads to a unified approach for substantially improving the running times of parameterized algorithms for some classic problems. This includes, among others, {\sc $k$-Partial Cover}, {\sc $k$-Internal Out-Branching}, and {\sc Long Directed Cycle}. Our approach exploits an interesting tradeoff between running time and the size of the representative families.

\end{abstract}

\vspace{-0.8em}

\mysection{Introduction}

\vspace{-0.132em}

The theory of matroids is unique in the extent to which it connects such disparate branches
of combinatorial theory and algebra as graph theory,
combinatorial optimization, linear algebra, and algorithm theory.
Marx \cite{marx09} was the first to apply matroids to design fixed-parameter tractable algorithms.
The main tool used by Marx was the notion of representative families. Representative
families for set systems were introduced by Monien \cite{monien85}.

\comment{
A parameterized algorithm with parameter $k$ is an algorithm that runs in time $O^*(f(k))$ for some function $f$, where $O^*$ hides factors polynomial in the input size. A fast computation of representative families for uniform matroids
plays a pivotal role in obtaining better running times for parameterized algorithms.
Indeed, after
each stage, in which the algorithm computes a family $\cal S$ of sets that are partial solutions,\footnote{Typically, this is done in algorithms that use dynamic programming.} we compute a subfamily $\widehat{\cal S}\!\subseteq\!{\cal S}$ that represents $\cal S$. Then, each reference to ${\cal S}$ can be replaced by a reference to $\widehat{\cal S}$. The representative family $\widehat{\cal S}$
 contains ``enough" sets from ${\cal S}$; therefore, such replacement preserves the correctness of the algorithm. Thus, if we can compute fast representative families that are small enough, we can substantially improve the running time of the algorithm.
}
Let $E$ be a universe of $n$ elements, and ${\cal I}$ a family of subsets of size at most $k$ of $E$, for some $k\!\in\!\mathbb{N}$, i.e.,  ${\cal I}\!\subseteq\!\{S\!\subseteq\!E\!:\!|S|\!\leq\!k\}$. Then, $U_{n,k}\!=\!(E,\!{\cal I})$ is called a {\em uniform matroid}. Consider such a matroid and a family ${\cal S}$ of $p$-subsets of $E$, i.e., sets of size $p$. A subfamily $\widehat{\cal S}\!\subseteq\!{\cal S}$ {\em represents} ${\cal S}$ if for every pair of sets $X\!\in\!{\cal S}$ and $Y\!\subseteq\!E\setminus\! X$, such that $X\!\cup\!Y\!\in\!{\cal I}$ (i.e., $|Y|\!\leq\!(k\!-\!p)$), there is a set $\widehat{X}\!\in\!\widehat{\cal S}$ disjoint from $Y$.
In other words, if a set $Y$ can be extended to an independent set
(of size at most $k$) by adding a subset from $S$, then it can also be extended to an independent set
(of the same size) by adding a subset from $\widehat{\cal S}$.

The {\em Two Families Theorem} of Bollob$\acute{\mathrm{a}}$s \cite{bollobas65} implies that for any uniform~matroid $U_{n,k}\!=\!(E,\!{\cal I})$ and a family ${\cal S}$ of $p$-subsets of $E$, for some $1\!\leq\!p\!\leq\!k$, there~is~a subfamily $\widehat{\cal S}\!\subseteq\! {\cal S}$ of size ${k \choose p}$ that represents ${\cal S}$. For more general matroids, the generalization of Lov$\acute{\mathrm{a}}$sz for this theorem, given in \cite{lovasz77}, implies a similar result, and algorithms based on this generalization are given in \cite{marx09} and~\cite{representative}.

A parameterized algorithm with parameter $k$ has running time $O^*(f(k))$ for some function $f$, where $O^*$ hides factors polynomial in the input size. A fast computation of representative families for uniform matroids
plays a central role in obtaining better running times for such algorithms.
Indeed, after
each stage, in which the algorithm computes a family $\cal S$ of sets that are partial solutions,\footnote{Typically, this is done in algorithms that use dynamic programming.} we compute a subfamily $\widehat{\cal S}\!\subseteq\!{\cal S}$ that represents $\cal S$. 
Then, each reference to ${\cal S}$ can be replaced by a reference to $\widehat{\cal S}$. The representative family $\widehat{\cal S}$
 contains ``enough" sets from ${\cal S}$; therefore, such replacement preserves the correctness of the algorithm. Thus, if we can compute fast representative families that are small enough, we can substantially improve the running time of the algorithm.

For uniform matroids, Monien \cite{monien85} computed representative families of size $\sum_{i\!=\!0}^{k\!-\!p}p^i$ in time $O(|{\cal S}|p(k\!-\!p)\sum_{i\!=\!0}^{k\!-\!p}p^i)$, and Marx \cite{marx06} computed representative families of size ${k \choose p}$ in time $O(|{\cal S}|^2p^{k\!-\!p})$. Recently, Fomin et al. \cite{representative} introduced a powerful technique
which enables to
compute representative
 families of size ${k \choose p}2^{o(k)}\!\log n$ in time $O(|{\cal S}|(k/(k\!-\!p))^{k\!-\!p}2^{o(k)}\!\log n)$, thus significantly improving the previous results.

In this paper, we show that the technique of \cite{representative} leads to a unified tradeoff-based approach for substantially improving the running time of parameterized algorithms for some classic problems. In particular, we demonstrate the applicability of our approach, for the following problems, among others (see also Section \ref{sec:results}).


\vspace{-0.2em}

\myparagraph{$k$-Partial Cover ($k$-PC)} Given a universe $U$, a family $\cal S$ of subsets of $U$ and a parameter $k\!\in\!\mathbb{N}$, find the smallest number of sets in ${\cal S}$ whose union contains at least $k$ elements.

\vspace{-0.2em}

\myparagraph{$k$-Internal Out-Branching ($k$-IOB)} Given a {\em directed} graph $G\!=\!(V,\!E)$ and a parameter $k\!\in\!\mathbb{N}$, decide if $G$ has an {\em out-branching} (i.e., a spanning tree having exactly one node of in-degree 0) with at least $k$ nodes of out-degree $\geq 1$.

\vspace{-0.8em}

\mysubsection{Prior Work}

\vspace{-0.2em}

The $k$-PC problem generalizes the well-known {\sc $k$-Dominating Set ($k$-DS)} problem, defined as follows. Given a graph $G\!=\!(V,\!E)$ and a parameter $k\!\in\!\mathbb{N}$, find the smallest size of a set $U\!\subseteq\!V$ such that the number of nodes that belong to $U$ or are neighbors of nodes in $U$ is at least $k$. If {\sc $k$-PC} can be solved in time $t(|U|,\!|{\cal S}|,\!k)$, then {\sc $k$-DS} can be solved in time $t(|V|,\!|V|,\!k)$ (see, e.g., \cite{bonnet13}). Note that the special cases of {\sc $k$-PC} and {\sc $k$-DS} in which $k\!=\!n$, are the classical NP-complete {\sc Set Cover} and {\sc Dominating Set} problems \cite{subgraphisoneg}, respectively. Table~\ref{tab:knownresults2} presents a summary of known parameterized algorithms for {\sc $k$-PC} and {\sc $k$-DS}. We note that the parameterized complexity of {\sc $k$-PC} and {\sc $k$-DS} has been studied also~with respect to other parameters and for more restricted inputs (see, e.g., \cite{bonnet13,subexponentialcov,approxcover}).

\begin{table}[center]
\centering
\begin{tabular}{|l|c|c|c|}
	\hline
	Reference 		                      & Deterministic$\setminus$Randomized & Variant             & Running Time           \\ \hline \hline		
	Bonnet et al.~\cite{bonnet13}	  		& {\em det} 																 & $k$-PC          		 & $O^*(4^kk^{2k})$       \\ \hline

	Bl$\ddot{\mathrm{a}}$ser~\cite{blaser03} &  {\em rand} 														 & $k$-PC          		 & $O^*(5.437^k)$         \\ \hline	

	Kneis et al.~\cite{sofsem07}				& {\em det} 																 & $k$-DS          		 & $O^*((16+\epsilon)^k)$ \\
																			&  {\em rand} 																 & $k$-DS          		 & $O^*((4+\epsilon)^k)$   \\ \hline

	Chen et al.~\cite{chenchen13}			  & {\em det} 																 & $k$-DS          		 & $O^*(5.437^k)$          \\ \hline

	Kneis~\cite{kneisthesis}						& {\em det} 																 & $k$-DS          		 & $O^*((4+\epsilon)^k)$  \\ \hline
	
	Koutis et al.~\cite{appmultilinear} & {\em rand}																   & $k$-DS 						 & $O^*(2^k)$             \\ \hline

	{\bf This paper}                    & {\bf det}                      			 & $\bf k${\bf-PC}     & $\bf O^*(2.619^k)$     \\ \hline
	
\end{tabular}\smallskip
\caption{Known parameterized algorithms for {\sc $k$-PC} and {\sc $k$-DS}.}
\label{tab:knownresults2}
\end{table}

The {\sc $k$-IOB} problem is of interest in database systems \cite{outbranchpatent}. A special case of {\sc $k$-IOB}, called {\sc $k$-Internal Spanning Tree ($k$-IST)}, asks if a given {\em undirected} graph $G\!=\!(V,\!E)$ has a spanning tree with at least $k$ internal nodes. An interesting application of {\sc $k$-IST}, for connecting cities with water pipes, is given in \cite{kISPbounddeg}. The {\sc $k$-IST} problem is NP-complete, since it generalizes the {\sc Hamiltonian Path} problem \cite{hamil3}; thus, {\sc $k$-IOB} is also NP-complete. Table~\ref{tab:knownresults1} presents a summary of known parameterized algorithms for {\sc $k$-IOB} and {\sc $k$-IST}. More details on {\sc $k$-IOB},

{\noindent{\sc $k$-IST} and variants of these problems can be found in the excellent surveys of \cite{iobsurv1,iobsurv2}.}

\begin{table}[center]
\centering
\begin{tabular}{|l|c|c|c|}
	\hline
	Reference 		                    & Deterministic$\setminus$Randomized & Variant          & Running Time          \\ \hline \hline		
	Priesto et al.~\cite{kISP24klogk}	& {\em det} 																 & $k$-IST          & $O^*(2^{O(k\log k)})$ \\ \hline	
	Gutin al.~\cite{kIOB2klogk}	      & {\em det} 																 & $k$-IOB          & $O^*(2^{O(k\log k)})$ \\ \hline
	Cohen et al.~\cite{kIOB49k}	      & {\em det} 																 & $k$-IOB          & $O^*(55.8^k)$         \\		
	 														      & {\em rand} 																 & $k$-IOB          & $O^*(49.4^k)$         \\ \hline	
	Fomin et al.~\cite{kIOB16k}	      & {\em det} 																 & $k$-IOB          & $O^*(16^{k+o(k)})$		\\ \hline	
	Fomin et al.~\cite{kISP8k}	      & {\em det} 																 & $k$-IST          & $O^*(8^k)$            \\ \hline	
	Zehavi~\cite{ipec13}					 	  & {\em rand} 																 & $k$-IOB          & $O^*(4^k)$            \\ \hline
	{\bf This paper}                  & {\bf det}                      			 & $\bf k${\bf-IOB} & $\bf O^*(6.855^k)$    \\ \hline
	
\end{tabular}\smallskip
\caption{Known parameterized algorithms for {\sc $k$-IOB} and {\sc $k$-IST}.}
\label{tab:knownresults1}
\end{table}

\vspace{-0.2em}

\mysubsection{Our Results}
\label{sec:results}

Given a uniform matroid $U_{n, k}\!=\!(E,{\cal I})$ and a family ${\cal S}$ of $p$-subsets of $E$, we compute a subfamily $\widehat{\cal S}\!\subseteq\!{\cal S}$ of size $\displaystyle{\frac{(ck)^{k}}{p^p(ck\!-\!p)^{k\!-\!p}}2^{o(k)}\log n}$ which represents ${\cal S}$, in time $\displaystyle{O(|{\cal S}|((ck)/(ck\!-\!p))^{k\!-\!p}}2^{o(k)}\log n)$, for any fixed $c\!\geq\!1$.
Indeed, taking $c\!=\!1$, we have 
the result of
Fomin et al. \cite{representative}.
As $c$ grows larger,
the size of $\widehat{\cal S}$ increases, with a corresponding decrease in computation time.
This enables to obtain better running times for the algorithms for {\sc Long Directed Cycle}, {\sc Weighted $k$-Path} and {\sc Weighted $k$-Tree}, as given in \cite{representative}.

In particular, we use this approach to develop {\em deterministic} algorithms that solve {\sc $k$-PC} and {\sc $k$-IOB} in times $O^*(2.619^k)$ and $O^*(6.855^k)$, respectively. We~thus significantly improve the {\em randomized} algorithm with the best known $O^*(5.437^k)$ running time for {\sc $k$-PC} \cite{blaser03}, and the deterministic algorithm with the best known $O^*(16^{k+o(k)})$ running time for {\sc $k$-IOB} \cite{kIOB16k}. This also improves the running times of the best known deterministic algorithms for {\sc $k$-DS} and {\sc $k$-IST} (see Tables \ref{tab:knownresults2} and \ref{tab:knownresults1}).

\myparagraph{Technical Contribution} Our unified approach exploits an interesting tradeoff between running time and the size of the representative families. This tradeoff is made precise by using, along with the scheme of \cite{representative}, a parameter $c\!\geq\! 1$, which enables a more careful selection of elements to the sets.

Indeed, towards computing a representative family $\widehat{\cal S}$, we seek a family ${\cal F}\!\subseteq\!2^E$ that satisfies the following condition. For every pair of sets $X\!\in\!{\cal S}$, and $Y\!\subseteq\!E\!\setminus\! X$ such that $X\!\cup\! Y\!\in\!{\cal I}$, there is a set $F\!\in\!\cal F$ such that $X\!\subseteq\! F$, and $Y\!\cap\! F\!=\!\emptyset$ (see Fig.~\ref{fig:1}). Then, we compute $\widehat{\cal S}$ by iterating over all $S\!\in\!{\cal S}$ and $F\!\in\!{\cal F}$ such that $S\!\subseteq\!F$. The time complexity of this iterative process is the dominant factor in the overall running time. Thus, we seek a small family $\cal F$, such that for any $S\!\in\!{\cal S}$, the expected number of sets in $\cal F$ containing $S$ is small. In constructing each set $F\!\in\!{\cal F}$, we insert each element $e\!\in\!E$ to $F$ with probability $p/(ck)$. For $c\!=\!1$, we get the probability used in \cite{representative}. When we take a larger value for $c$, we need to construct a larger family $\cal F$ in order to satisfy the above condition. Yet, since elements in $E$ are inserted to sets in $\cal F$ with a smaller probability, we get that for any $S\!\in\!{\cal S}$, the expected number of sets in $\cal F$ containing $S$ is smaller.

\begin{figure}
\medskip
\centering
\frame{
\includegraphics[scale=0.29]{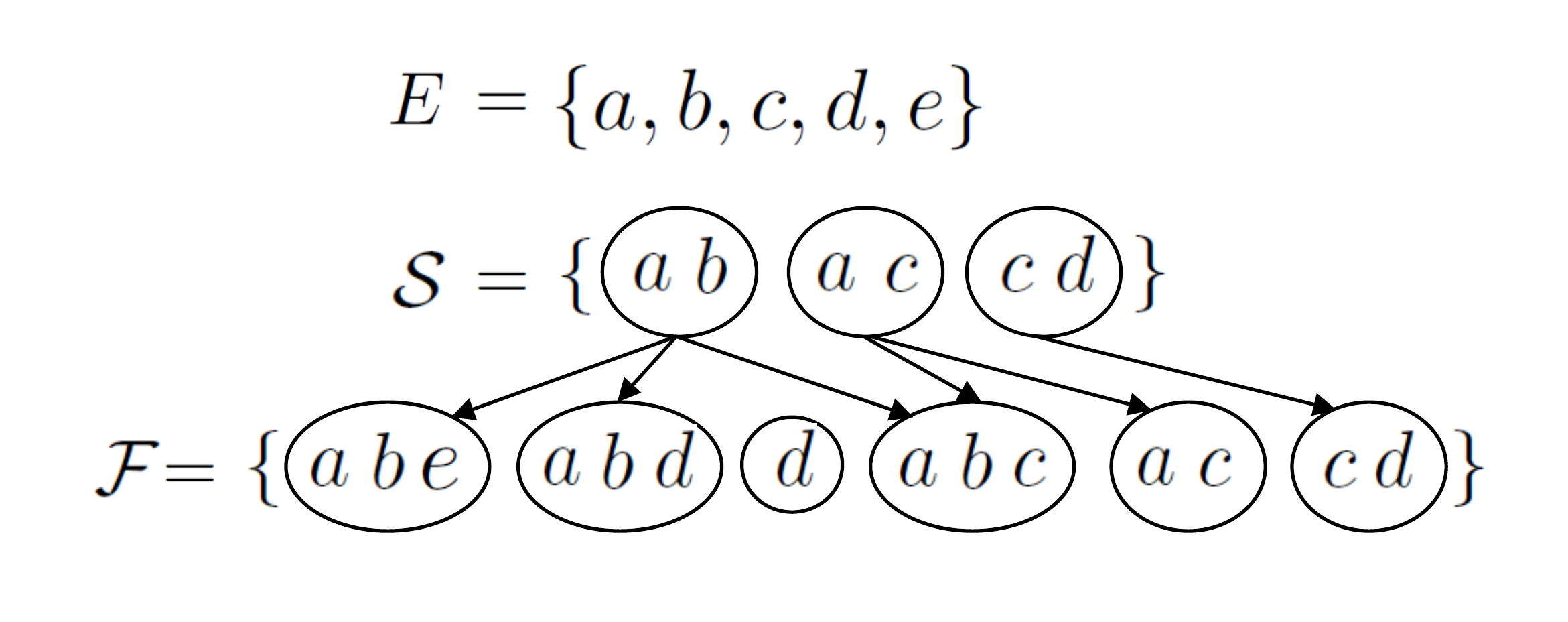}
}
\caption{An illustration of a family ${\cal F}\!\subseteq\!2^E$. Assume that $n\!=\!5,k\!=\!4$, and $p\!=\!2$. An arrow from a set $S\!\in\!{\cal S}$ to a set $F\!\in\!{\cal F}$ indicates that $S\!\subseteq\!F$.}
\label{fig:1}
\medskip
\end{figure}

\myparagraph{Organization} Section \ref{section:preliminaries} gives some definitions and notation. Section \ref{section:computation} presents a tradeoff between running time and the size of the representative families. Using this computation, we derive in Sections \ref{section:pc} and \ref{section:kiob} our main results, which are fast parameterized algorithms for {\sc $k$-PC} and {\sc $k$-IOB}. Finally, Section \ref{section:knownapp} shows the improvements in running times resulting from our tradeoff-based approach for three previous applications of representative families of \cite{representative}.

\mysection{Preliminaries}\label{section:preliminaries}

\vspace{-0.2em}

We first define the weighted version of representative families.

\vspace{-0.2em}

\begin{definition}\label{def:repfam}
Given a matroid $U_{n,k}\!=\!(E,\!{\cal I})$, a family ${\cal S}$ of $p$-subsets of $E$, and a function $w:\!{\cal S}\!\rightarrow\!\mathbb{R}$, we say that a subfamily $\widehat{\cal S}\!\subseteq\!{\cal S}$ {\em max (min) represents} $\cal S$ if for every pair of sets $X\in{\cal S}$, and $Y\!\subseteq\!E\!\setminus\!X$ such that $X\cup Y\in{\cal I}$, there is a set $\widehat{X}\!\in\!\widehat{\cal S}$ disjoint from $Y$ such that $w(\widehat{X})\!\geq\!w(X)$ $(w(\widehat{X})\!\leq\!w(X))$.
\end{definition}
We give an illustration of a representative family in Fig.~\ref{fig:2} (see Appendix~\ref{app:illustration}).~The special case where $w(\!S\!)\!=\!0$, for all $S\!\in\!{\cal S}$, is the unweighted version of Definition~\ref{def:repfam}.

The following simple observation (used in Sections \ref{section:pc} and \ref{section:kiob}) asserts that representation is a transitive relation among families of subsets.

\vspace{-0.2em}

\begin{obs}[\cite{representative}]\label{obs:transitive}
Let $U_{n,k}\!=\!(E,\!{\cal I})$ be a matroid. Let ${\cal S},\!{\cal T}$ and ${\cal R}$ be families of $p$-subsets of $E$, and let $w$ be a function from ${\cal S}\cup{\cal T}\cup{\cal R}$ to $\mathbb{R}$. If ${\cal S}$ max (min) represents ${\cal T}$ and ${\cal T}$ max (min) represents ${\cal R}$, then ${\cal S}$ max (min) represents $\cal R$.
\end{obs}

\vspace{-0.2em}

\myparagraph{Notation} Given a set $U$ and a nonnegative integer $t$, let ${U \choose t}\!=\!\{U'\!\subseteq\! U\!:\! |U'|\!=\!t\}$. Also, recall that an out-tree $T$ is a directed tree having exactly one node of in-degree 0, called {\em the root}. We denote by $V_T$, $E_T$, $i(T)$ and $\ell(T)$ the node set, edge set, number of internal nodes (i.e., nodes of out-degree $\geq\!1$) and number of leaves (i.e., nodes of out-degree 0), respectively.

\mysection{A Tradeoff-Based Approach}\label{section:computation}

In this section we prove the following result.

\vspace{-0.2em}

\begin{theorem}\label{theorem:newrep}
Given a parameter $c\!\geq\!1$, a uniform matroid $U_{n,k}\!=\!(E,\!{\cal I})$, a family ${\cal S}$ of $p$-subsets of $E$, and a function $w\!:\!{\cal S}\!\rightarrow\!\mathbb{R}$,  a family $\widehat{\cal S}\!\subseteq\!{\cal S}$ of size $\displaystyle{\frac{(ck)^{k}}{p^p(ck\!-\!p)^{k\!-\!p}}\!2^{o(k)}\!\log\!n}$~that max (min) represents $\cal S$ can be found in time $\displaystyle{O\!(|{\cal S}|(ck/(ck\!-\!p))^{k\!-\!p}2^{o(k)}\!\log\!n\! +\! |{\cal S}|\!\log\!|{\cal S}|)}$.
\end{theorem}
Note that, in the special case where $c\!=\!1$, we have the statement of Theorem 6 in \cite{representative}. Roughly speaking, the proof of Theorem \ref{theorem:newrep} is structured as follows. We first argue that we can focus on finding a certain data structure to compute representative families. Then, we construct such a data structure that is not as efficient as required (first randomly, and then deterministically). This part contains our main contribution. Finally, we show how to improve the {\em ``efficiency"} of this data structure (this is made precise below).

\vspace{-0.1em}

\begin{proof}
Clearly, we may assume that $|{\cal S}|\geq \displaystyle{\frac{(ck)^{k}}{p^p(ck\!-\!p)^{k\!-\!p}}2^{o(k)}\!\log\!n}$. Recall that our computation of representative families requires finding initially a family ${\cal F}\!\subseteq\!2^E$ that satisfies the following condition. For every pair of sets $X\!\in\!{\cal S}$, and $Y\!\subseteq\!E\!\setminus\!X$, such that $X\!\cup\!Y\!\in\!{\cal I}$, there is a set $F\!\in\!{\cal F}$ such that $X\!\subseteq\! F$, and $Y\!\cap\! F\!=\!\emptyset$ (see Fig.~\ref{fig:1}). An {\em $(n,\!k,\!p)$-separator} is a data structure containing such a family ${\cal F}$, which, given a set $S\!\in\!{E \choose p}$, outputs the subfamily of sets in $\cal F$ that contain $S$, i.e., $\chi(S)\!=\!\{F\!\in\!{\cal F}\!: S\!\subseteq\!F\}$.

To derive our fast computation, we need to construct an efficient $(\!n,\!k,\!p\!)$-separator, where efficiency is measured by the following parameters: ${\cal C}\!=\!{\cal C}(n,\!k,\!p)$, the number of sets in the family $\cal F$; $\tau_{\cal F}\!=\!\tau_{\cal F}(n,\!k,\!p)$, the time required to compute the family ${\cal F}$; $\Delta\!=\!\Delta(n,\!k,\!p)$, the maximum size of $\chi(S)$, for any $S\!\in\! {E \choose p}$, and~$\tau_{\chi}\!=$\\$\tau_{\chi}(n,\!k,\!p)$, an upper bound for the time required to output $\chi(S)$, for any $S\!\in\!{E \choose p}$.

Given such a separator, a subfamily $\widehat{\cal S}\!\subseteq\! {\cal S}$ of size ${\cal C}$ that max $(\!\mathrm{min}\!)$ represents ${\cal S}$ can be constructed in time $O(\!\tau_{\cal F} \!+\! |{\cal S}|\tau_{\chi} \!+ \!|{\cal S}|\!\log\!|{\cal S}|\!)$ as follows.~First,~compute~$\cal F$, and $\chi(\!S\!)$ for all $S\!\in\!{\cal S}$.~Then, order ${\cal S}\!=\!\{\!S_1,\!...,\!S_{|{\cal S}|}\!\}$, such that $w(\!S_{i\!-\!1}\!)\!\geq\! w\!(\!S_i\!)$~$(\!w\!(\!S_{i\!-\!1}\!)$\\$\!\leq\! w(S_i))$, for all $2\!\leq\! i\!\leq\! |{\cal S}|$. Finally, return all $S_i\!\in\!{\cal S}$ for which there is a set $F\!\in\!{\cal F}$ containing $S_i$ but no $S_j$, for $1\!\leq\!j\!<\!i$. Formally, return $\widehat{\cal S}\!=\!\{S_i\!\in\!{\cal S}\!:\! \chi(S_i)\!\setminus\!(\bigcup_{1\leq j<i}\chi(S_j))\!\neq\!\emptyset\}$. The correctness of this construction is proved in \cite{representative}. Thus, to prove the theorem it suffices to find an $(n,\!k,\!p)$-separator with~parameters:

\vspace{-0.35em}
\begin{itemize}
\item $\displaystyle{{\cal C}^* \leq \frac{(ck)^{k}}{p^p(ck-p)^{k-p}}2^{o(k)}\log n}$.
\item $\displaystyle{\tau_{\cal F}^* \leq \frac{(ck)^{k}}{p^p(ck-p)^{k-p}}2^{o(k)}n\log n}$.
\item $\displaystyle{\tau_{\chi}^* \leq (ck/(ck-p))^{k-p}2^{o(k)}\log n}$.
\end{itemize}
We start by giving an $(n,\!k,\!p)$-separator, that we call Separator 1, with the following parameters, which are worse than required:

\vspace{-0.35em}
\begin{itemize}
\item $\displaystyle{{\cal C}^1=O(\frac{(ck)^{k}}{p^p(ck-p)^{k-p}}k^{O(1)}\log n)}$.
\item $\displaystyle{\tau_{\cal F}^1=O({2^n\choose {\cal C}^1}n^{O(k)})}$.
\item $\displaystyle{\Delta^1=O((ck/(ck-p))^{k-p}k^{O(1)}\log n)}$.
\item $\displaystyle{\tau_{\chi}^1=O(\frac{(ck)^{k}}{p^p(ck-p)^{k-p}}n^{O(1)})}$.
\end{itemize}
First, we give a randomized algorithm which constructs, with positive probability, an $(n,\!k,\!p)$-separator having the desired ${\cal C}^1$ and $\Delta^1$ parameters. We then show how to deterministically construct an $(n,\!k,\!p)$-separator having all the desired parameter values. Let $\displaystyle{t\!=\!\frac{(ck)^{k}}{p^p(ck-p)^{k-p}}(k\!+\!1)\!\ln\!n}$, and construct the family ${\cal F}\!=\!\{F_1,\!\ldots,\!F_t\}$ as follows. For each $i\!\in\!\{1,\!\ldots,\!t\}$ and element $e\!\in\!E$, insert $e$ to $F_i$ with probability $p/(ck)$. The construction of different sets in ${\cal F}$, as well as the insertion of different elements into each set in ${\cal F}$, are independent. Clearly, ${\cal C}^1\!=\!t$ is within the required bound.

For fixed sets $X\!\in\!{E \choose p}$, $Y\!\in\!{E\setminus X\choose k-p}$ and $F\!\in\!{\cal F}$, the probability that $X\!\subseteq\!F$ and $Y\!\cap\! F\!=\!\emptyset$ is $\displaystyle{(\frac{p}{ck})^p(1\!-\!\frac{p}{ck})^{k\!-\!p}\! = \!\frac{p^p(ck\!-\!p)^{k\!-\!p}}{(ck)^{k}}\! =\! (k\!+\!1)\!\ln\!n/t}$. Thus, the probability that {\em no} set $F\!\in\!{\cal F}$ satisfies $X\subseteq F$ and $Y\cap F=\emptyset$ is $\displaystyle{(1\! -\! (k\!+\!1)\!\ln n/t)^t\!\leq\! e^{-(k+1)\!\ln\!n}\!}$\\$=\!n^{-k-1}$. There are at most $n^k$ choices for $X\!\in\!{E \choose p}$ and $Y\!\in\!{E\setminus X\choose k-p}$; thus, applying the union bound, the probability that there exist $X\!\in\!{E \choose p}$ and $Y\!\in\!{E\setminus X \choose k-p}$, such that no set $F\!\in\!{\cal F}$ satisfies $X\!\subseteq\!F$ and $Y\!\cap\!F\!=\!\emptyset$, is at most $n^{-k-1}\cdot n^k=\frac{1}{n}$.

For any sets $S\!\in\!{E \choose p}$ and $F\!\in\!{\cal F}$, the probability that $S\!\subseteq\!F$ is $(p/(\!ck\!))^p$.~Therefore, $|\chi(\!S\!)|$, the number of sets in $\cal F$ containing $S$, is a sum of $t$ i.i.d.~Bernoulli random variables with parameter $(p/(\!ck\!))^p$. Then, the expected value of $|\chi(\!S\!)|$ is $\displaystyle{E[|\chi(\!S\!)|]\!=\!t(\!\frac{p}{ck}\!)^p}\!=\!(\!\frac{ck}{ck\!-\!p}\!)^{k\!-\!p}\!(\!k\!+\!1\!)\!\ln n$. Applying standard Chernoff~bounds, we have that the probability that $|\chi(\!S\!)|\!\geq\!6E[|\chi(\!S\!)|]$ is upper bounded by $2^{-\!6E[|\chi(\!S\!)|]}$\\$\leq\! n^{-\!k-1}$. There are ${n\choose p}$ choices for $S\!\in\!{E \choose p}$. Thus, by the union bound, the probability that $\displaystyle{\Delta^1 \!>\! 6\!\cdot\![((\!ck\!)/(\!ck\!-\!p\!))^{k\!-\!p}(\!k\!+\!1\!)\!\ln n]}$ is upper bounded by $\frac{1}{n}$.

So far, we have given a randomized algorithm that constructs an $(n,\!k,\!p)$-separator having the desired ${\cal C}^1$ and $\Delta^1$ parameters with probability at least $1\!-\!\frac{2}{n}\!>\!0$. To deterministically construct $\cal F$ in time bounded by $\tau_{\cal F}^1$, we iterate over all families of $t$ subsets of $E$ (there are ${2^n\choose {\cal C}^1}$ such families), where for each family $\cal F$, we test in time $n^{O(k)}$ whether $\Delta^1$ is within the required bound, and whether for any pair of sets $X\!\in\!{E \choose p}$ and $Y\!\in\!{E\setminus X\choose k-p}$, there is a set $F\!\in\!{\cal F}$ such that $X\!\subseteq\!F$ and $Y\!\cap\!F\!=\!\emptyset$. Then, given a set $S\!\in\!{E \choose p}$, we can deterministically compute $\chi(S)$ within the stated bound for $\tau_{\chi}^1$, by iterating over ${\cal F}$ and inserting each set that contains $S$.


To obtain an $(n,\!k,\!p)$-separator having the parameters ${\cal C}^*$, $\tau_{\cal F}^*$ and $\tau_{\chi}^*$, repeatedly apply Lemmas 4.4 and 4.5 of \cite{representative} to Separator 1 (see Appendix \ref{app:theoremcont}).\qed
\end{proof}
We note that the scheme for computing representative families, developed in the proof of Theorem \ref{theorem:newrep}, consists of three main stages:
\vspace{-0.1em}
\begin{enumerate}
\item Construct several ``small" inefficient separators.
\item Given the results of Stage 1, construct an efficient $(n,k,p)$-separator.
\item Use the separator generated in Stage 2 to construct a representative family.
\end{enumerate}
\vspace{-0.1em}
{\noindent We give a pseudocode of the scheme, called \alg{RepAlg}, in Appendix \ref{app:scheme}.}

\mysection{An Algorithm for {\sc $k$-Partial Cover}}\label{section:pc}

We now show how to apply our scheme, \alg{RepAlg}, to obtain a faster parameterized algorithm for {\sc $k$-PC}. Let $m\!=\!|{\cal S}|$ be the number of sets in ${\cal S}$. The main idea of the algorithm is to iterate over the sets in ${\cal S}$ in some arbitrary~order $S_1,\!S_2,\!...,\!S_m$, such that when we reach a set $S_i$, we have already computed representative families for families of ``partial solutions" that include only elements from the sets $S_1\!,\!...,\!S_{i\!-\!1}$.~Then, we try to extend the partial solutions by adding {\em uncovered} elements from $S_i$.~The key observation, that leads to our improved running time, is that we cannot simply add ``many" elements from $S_i$ at once, but rather add these elements one-by-one; thus, we can compute new representative families after adding each element, which are then used in adding the next element.

\myparagraph{The Algorithm} We now describe  \alg{PCAlg}, our algorithm for {\sc $k$-PC} (see the pseudocode below). The first step solves the simple case where the solution is '1'. Then, algorithm \alg{PCAlg} generates a matrix M, where each entry M$[i,\!j,\!\ell]$ holds a family that represents  $Sol_{i,j,\ell}$, the family of partial solutions including $j$ elements, obtained from a subfamily ${\cal S}'$ of $\ell$ sets in $\{S_1,\!...,S_i\}$, i.e., $Sol_{i,j,\ell} = \{S\!\subseteq\!(\bigcup{\cal S'}): {\cal S'}\!\subseteq\!\{S_1,\!...,S_i\}, |S|\!=\!j, |{\cal S'}|\!=\!\ell\}$.

Algorithm \alg{PCAlg} iterates over all triples $(i,\!j,\!\ell)$, where $i\!\in\!\{\!1,\!...,\!m\}, j\!\in\!\{\!1,\!...,\!k\}$ and $\ell\!\in\!\{\!1,\!...,\min\{i,\!k\}\}$. In each iteration, corresponding to a triple $(\!i,\!j,\!\ell\!)$, \alg{PCAlg} computes M$[i,\!j,\!\ell]$, by using M$[i\!-\!1,\!j',\!\ell\!-\!1]$, for all $1\!\leq\! j'\!\leq\! j$, and M$[i\!-\!1,\!j,\!\ell]$. In other words, \alg{PCAlg} computes a family that represents $Sol_{i,j,\ell}$ by using families that represent $Sol_{i\!-\!1,j',\ell\!-\!1}$, for all $1\!\leq j'\!\leq j$, and $Sol_{i\!-\!1,j,\ell}$. In particular, algorithm \alg{PCAlg} adds elements in $S_i$ one-by-one to sets in M$[i\!-\!1,\!j',\!s\!-\!1]$, for all $1\!\leq\! j'\!\leq\! j$. After adding an element, \alg{PCAlg} computes (in Step \ref{pcalg:new3}) new representative families, to be used in adding the next element. Let $S_i\!=\!\{s_1,\!...,s_r\}$. Then, \alg{PCAlg} computes a family ${\cal A}_{r',\!j'}$, for all $1\!\leq\! r'\!\leq\! r$ and $0\!\leq\! j'\!\leq\! j$, that represents the family of partial solutions including $j'$ elements, obtained from $\{s_1,\!...,s_{r'}\}$ and $(\ell\!-\!1)$ sets in $\{S_1,\!...,\!S_{i\!-\!1}\}$. The family ${\cal A}_{r',\!j'}$ is computed by calling \alg{RepAlg} with the {\em family parameter} containing the union of ${\cal A}_{r'\!-\!1,\!j'}$ and the family of sets obtained by adding $s_{r'}$ to sets in ${\cal A}_{r'\!-\!1,\!j'\!-\!1}$.

Suppose the solution is $\ell^*$. Then, using representative families guarantees that each entry M$[i,\!j,\!\ell]$ holds ``enough" sets from $Sol_{i,j,\ell}$, such that when the algorithm terminates, M$[m,\!k,\!\ell^*]\!\neq\!\emptyset$. Moreover, using representative families guarantees that each entry M$[i,\!j,\!\ell]$ does not hold ``too many" sets from $Sol_{i,j,\ell}$, thereby yielding an improved running time.

\setlength{\textfloatsep}{1em}

\begin{algorithm}[!ht]
\caption{\alg{PCAlg}($U,k,{\cal S}=\{S_1,\ldots,S_m\}$)}
\begin{algorithmic}[1]\label{alg:cov}
\STATE\label{pcalg:1} {\bf if} there is $S\in{\cal S}$ s.t. $|S|\geq k$ {\bf then} {\bf return} 1. {\bf end if}

\STATE\label{pcalg:2} let M be a matrix that has an entry $[i,j,\ell]$ for all $0\!\leq\! i\!\leq\! m,1\!\leq\! j\!\leq\! k$ and $0\!\leq\! \ell\!\leq\! k$, initialized to $\emptyset$.

\FOR{$i=1,\ldots,m,\ j=1,\ldots,k,\ \ell=1,\ldots,\min\{i,\!k\}$}\label{pcalg:ite}	
	\STATE\label{pcalg:new0} let $S_i=\{s_1,\ldots,s_r\}$.
	
	
	\STATE\label{pcalg:new1} $\displaystyle{{\cal A}_{0,0}\!\Leftarrow\! \{\emptyset\}}$, and {\bf for} $j'\!=\!1,\!\ldots,\!j$ {\bf do} $\displaystyle{{\cal A}_{0,\!j'}\!\Leftarrow\! \mathrm{M}[i\!-\!1,\!j',\!\ell\!-\!1]}$. {\bf end for}
	
	\FOR{$r'\!=\!1,\ldots,r$, $j'\!=\!0,\ldots j$}\label{pcalg:new2}
		\STATE\label{pcalg:new3} ${\cal A}_{r',j'}\Leftarrow\ $\alg{RepAlg}$\displaystyle{(U,k,{\cal A}_{r'\!-\!1,j'}\cup\{S\!\cup\!\{s_{r'}\}: j'\!\geq\! 1, \ S\!\in\! {\cal A}_{r'\!-\!1,j'\!-\!1}, \ s_{r'}\!\notin\! S\})}$.
	\ENDFOR\label{pcalg:new4}
	

	
	\STATE\label{pcalg:4} $\displaystyle{\mathrm{M}[i,j,\ell]\Leftarrow}$ \alg{RepAlg}$\displaystyle{(U,k,\mathrm{M}[i-1,j,\ell]\cup{\cal A}_{r,j})}$.
\ENDFOR

	\STATE\label{pcalg:5} return the smallest $\ell$ such that M$[m,\!k,\!\ell]\!\neq\!\emptyset$.
\end{algorithmic}
\end{algorithm}

\myparagraph{Correctness and Running Time} We first state a lemma referring to Steps \ref{pcalg:new1}--\ref{pcalg:new4} in \alg{PCAlg} (we give the proof in Appendix \ref{app:lemmacoveradd}). In this lemma, we use the following notation. For all $0\!\leq\! i\!\leq\! m,1\!\leq\! j\!\leq\!k$ and $0\!\leq\! \ell\!\leq\! k$, we let ${\cal A}^*_{i,j,\ell}$ denote the family of sets containing $j$ elements, constructed by adding elements from $S_i$ to sets in $(\bigcup_{1\leq j'\leq j}\mathrm{M}[i\!-\!1,\!j',\!\ell\!-\!1])\!\cup\!\{\emptyset\}$, i.e., ${\cal A}^*_{i,j,\ell} = \{S\!\cup\! S'_i: S\!\in\! (\bigcup_{1\leq j'\leq j}\mathrm{M}[i\!-\!1,\!j',\!\ell\!-\!1])\!\cup\!\{\emptyset\}, S'_i\!\subseteq \!S_i, |S\!\cup\!S'_i|\!=\!j\}$.

\vspace{-0.2em}

\begin{mylemma}\label{lemma:coveradd}
Consider an iteration of Step \ref{pcalg:ite}	in \alg{PCAlg}, corresponding to some values $i,j$ and $\ell$. Then, the family ${\cal A}_{r,j}$ represents the family ${\cal A}^*_{i,j,\ell}$.
\end{mylemma}

\vspace{-0.1em}

{\noindent We use Lemma \ref{lemma:coveradd} in proving the next lemma, which shows the correctness~of~\alg{PCAlg} (see Appendix~\ref{app:coveralg}).}

\vspace{-0.25em}

\begin{mylemma}\label{lemma:coveralg}
For all $0\!\leq\! i\!\leq\! m$,$1\!\leq\! j\!\leq\!k$ and $0\!\leq\! \ell\!\leq\! k$, M$[i,j,\ell]$ represents $Sol_{i,j,\ell}$.
\end{mylemma}

\vspace{-0.15em}

{\noindent We summarize in the next result.}

\vspace{-0.15em}

\begin{theorem}\label{theorem:cover}
\alg{PCAlg} solves {\sc $k$-PC} in time $O(2.619^k|{\cal S}|\log^2|U|)$.
\end{theorem}

\vspace{-1em}

\begin{proof}
Lemma \ref{lemma:coveralg} and Step \ref{pcalg:5} imply that \alg{PCAlg} solves {\sc $k$-PC}. Also, Lemmas \ref{lemma:coveradd} and \ref{lemma:coveralg}, and the way \alg{RepAlg} proceeds, imply that \alg{PCAlg} runs in time

\vspace{-0.25em}

\[O(2^{o(k)}|{\cal S}|\log^2 |U|\cdot\max_{0\leq t\leq k}\left\{\frac{(ck)^{k}}{t^t(ck-t)^{k-t}}(\frac{ck}{ck-t})^{k-t}\right\})\]

\vspace{-0.25em}

{\noindent By choosing $c\!=\!1.447$, the maximum is obtained at $t\!=\!\alpha k$, where $\alpha\!\cong\!0.55277$. Thus, \alg{PCAlg} runs in time $O(2.61804^k|{\cal S}|\log^2|U|)$.}\qed
\end{proof}


\vspace{-0.25em}

\mysection{An Algorithm for {\sc $k$-Internal Out-Branching}}\label{section:kiob}

\vspace{-0.25em}

We show below how to use our scheme, \alg{RepAlg}, to obtain a faster parameterized algorithm for {\sc $k$-IOB}. We first define an auxiliary problem called {\sc $(k,\!t)$-Tree}, which requires finding a tree on a ``small" number of nodes, rather than a spanning tree. Given a directed graph $G\!=\!(\!V,\!E\!)$, a node $r\!\in\! V$, and nonnegative integers $k$ and $t$, the {\sc $(k,\!t)$-Tree} problem asks if $G$ contains an out-tree $T$ rooted at $r$, such that $i(T)\!=\!k$ and $\ell(T)\!=\!t$. The following lemma implies that we can focus on solving {\sc $(k,\!t)$-Tree}.

\begin{mylemma}[\cite{ipec13}]\label{lemma:tokltree}
If {\sc $(k,t)$-Tree} can be solved in time $\tau(G,\!k,\!t)$, then {\sc $k$-IOB} can be solved in time $O(|V|(|E|\!+\!\sum_{1\leq t\leq k}\tau(G,\!k,\!t)))$.
\end{mylemma}
Next, we show how to solve {\sc $(\!k,\!t\!)$-Tree}. Our solution technique is based~on~iterati-\\ng over all pairs of nodes $v,u\!\in\! V$, and all values $0\!\leq\! i\!\leq \!k\!-\!1$ and $0\!\leq\! \ell\!\leq\! t$. When we reach such $v,u,i$ and $\ell$, we have already computed, for all $v',u'\!\in\! V$, $0\!\leq i'\!\leq i$, and $0\!\leq\! \ell'\!\leq\! \ell$ satisfying $i'\!+\!\ell'\!<\!i\!+\ell\!$, representative families for families of sets that are ``partial solutions". Each such set contains nodes of an out-tree of $G$ that is rooted at $v'$, includes $u'$ as a leaf (unless $v'\!=\!u'$) and consists of $i'$ internal nodes (excluding $v'$) and $\ell'$ leaves (excluding $u'$). We then try to ``connect" out-trees represented by partial solutions in a manner that results in a {\em legal} out-tree---i.e., an out-tree of $G$ that is rooted at $v$, includes $u$ as a leaf (unless $v\!=\!u$) and consists of $i$ internal nodes (excluding $v$) and $\ell$ leaves (excluding $u$). In constructing a set of such legal out-trees, we add families of ``small" partial solutions one-by-one, so we can compute new representative families after adding each family, and then use them in adding the next one---this is a crucial point in obtaining our improved running time. The construction itself is quite technical. On a high level, it consists of iterating over some trees that indicate which families of partial solutions should be currently used, and in which order they should be added.

\myparagraph{Some Definitions} Let $d\!\geq\! 2$ be a constant $(\!\mathrm{to}$ be $\mathrm{determined}\!)$. Given nodes~$v,\!u\!\in\! V$, $0\!\leq\! i\!\leq\! k\!-\!1$ and $0\!\leq\! \ell\!\leq\! t$, let ${\cal T}_{v,u,i,\ell}$ be the set of out-trees of $G$ rooted at $v$, having exactly $i$ internal nodes and $\ell$ leaves, excluding $v$ and $u$, where $v\!=\!u$ or $u$ is a leaf. Also, let $Sol_{v,u,i,\ell} \!=\!\{V_T\!\setminus\!\{v,\!u\}\!:T\!\in\!{\cal T}_{v,u,i,\ell}\}$. Given nodes $v,\!u\!\in\! V$, let ${\cal C}_{v,u}$ be the set of trees $C$ rooted at $v$, where $v\!=\!u$ or $u$ is a leaf, $V_C \!\subseteq\! V$, and $3 \!\leq\! |V_C| \!\leq\! 4d$. Given a node $v$ of a rooted tree $T$, let $f_T\!(\!v\!)$ be the father of $v$ in $T$.

\myparagraph{The Algorithm} We now describe \alg{TreeAlg}, our algorithm for {\sc $(k,t)$-Tree} (see the pseudocode below). \alg{TreeAlg} first generates a matrix M, where each entry M$[v,\!u,\!i,\!\ell]$ holds a family that represents $Sol_{v,u,i,\ell}$. \alg{TreeAlg} iterates over all $v,u\in V$, $i\!\in\!\{0,\!...,k\!-\!1\}$, and $\ell\!\in\!\{0,\!...,t\}$ such that $1\!\leq\! i\!+\!\ell$. Next, consider some iteration, corresponding to such $v,\!u,\!i$ and $\ell$.

The goal in each iteration is to compute M$[v,\!u,\!i,\!\ell]$, by using entries that are already computed. Algorithm \alg{TreeAlg} generates a matrix N, where each entry N$[C]$ holds a family that represents the subfamily of $Sol_{v,u,i,\ell}$ including the node set (excluding $v$ and $u$) of each out-tree $T\!\in\! {\cal T}_{v,u,i,\ell}$ {\em complying} with the rooted tree $C$ as follows (see Fig.~\ref{fig:3}): (1) for any two nodes $v',\!u'\!\in\! V_C$, $v'$ is an ancestor of $u'$ in $C$ iff $v'$ is an ancestor of $u'$ in $T$, (2) the leaves in $C$ are leaves in $T$, and (3) in the forest obtained by removing $V_C$ from $T$, each tree has at most $(k\!+\!t)/d$ nodes and at most two neighbors in $T$ from $V_C$. Roughly speaking, each entry N$[C]$ is easier to compute than the entry M$[v,\!u,\!i,\!\ell]$, since $C$ ``guides" us through the computation as follows. The rooted tree $C$ implies which entries in M are relevant to N$[C]$, in which order they should be used, and, in particular, it ensures that these are only entries of the form M$[v',\!u',\!i',\!\ell']$, where $i'\!+\!\ell'\!\leq\!(k\!+\!t)/d$. This bound on $i'\!+\!\ell'$ ensures that the families for which we compute representative families are ``small", thereby reducing the running time of calls to \alg{RepAlg}. Next, consider an iteration corresponding to some $C\!\in\! {\cal C}_{v,u}$.

The current goal is to compute N$[C]$ by using the guidance of $C$, as we now describe in detail. Algorithm \alg{TreeAlg} generates a matrix L, where each entry L$[j,i',\ell']$ holds a family that represents the family of node sets, excluding nodes in $V_C$, of forests in ${\cal F}_{v,u,C,j,i',\ell'}$, which is defined as follows.
The set ${\cal F}_{v,u,C,j,i',\ell'}$ includes each subforest $F'$ of $G$ complying with the subforest $F$ of $C$ induced by $\{w_1,\!...,w_j\}$ in a manner similar to the above compliance of an out-tree $T\!\in\! {\cal T}_{v,u,i,\ell}$ with $C$, such that: (1) $V_{F'}\!\cap\!(V_C\setminus V_F)\!=\!\emptyset$, and (2) the number of internal nodes (leaves) in $F'$, excluding nodes in $V_F$, is $i'$ ($\ell'$).
Informally, we consider such a subforest $F'$ as a stage towards computing an out-tree $T\!\in\! {\cal T}_{v,u,i,\ell}$ that complies with $C$. Indeed, note that ${\cal F}_{v,u,C,|V_C|,i^*,\ell^*}$ is the set of out-trees in ${\cal T}_{v,u,i,\ell}$ that comply with $C$. Roughly speaking, the matrix L is computed by using dynamic programming and algorithm \alg{RepAlg} (in Steps \ref{step:tree9}--\ref{step:tree13}) as follows. Each entry in L is computed by adding node sets of certain trees to node sets of forests computed at a previous stage, and then calling algorithm \alg{RepAlg} to compute a representative family for the result.

\begin{figure}
\medskip
\centering
\frame{
\includegraphics[scale=0.35]{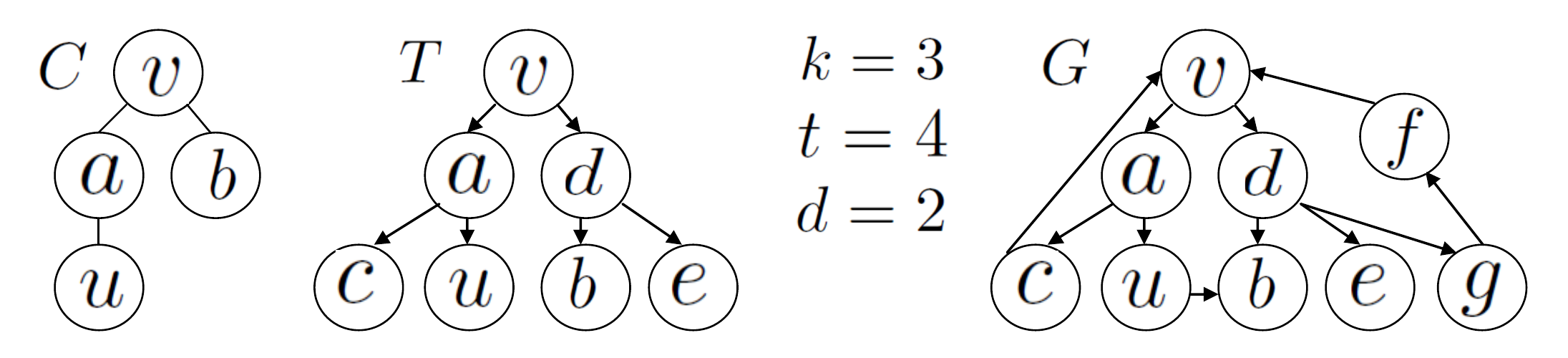}
}
\caption{An example for an out-tree $T\!\in\!{\cal T}_{v,\!u,\!2,\!3}$, complying with the rooted tree $C\!\in\!{\cal C}_{v,\!u}$.}
\label{fig:3}
\end{figure}

\renewcommand{\thealgorithm}{2}
\floatname{algorithm}{Algorithm}
\begin{algorithm}[!ht]
\caption{\alg{TreeAlg}($G\!=\!(V,E),r,k,t$)}
\begin{algorithmic}[1]\label{alg:tree}
\STATE\label{step:tree1} let M be a matrix that has an entry $[v,u,i,\ell]$ for all $v,u\in V$, $0\leq i\leq k-1$ and $0\leq \ell\leq t$, which is initialized to $\emptyset$.
\STATE\label{step:tree2} M$[v,u,0,0]\Leftarrow\{\emptyset\}$ for all $v,u\in V$ s.t. $(v,u)\in E$ or $v=u$.

\FOR{all $v,u\in V$, $i=0,\ldots,k-1$, $\ell=0,\ldots,t$ s.t. $1\leq i+\ell$}\label{step:tree3}
	\STATE\label{step:tree5} let N be a matrix that has an entry $[C]$ for all $C\in{\cal C}_{v,u}$.
	
	\FOR{all $C\in {\cal C}_{v,u}$}\label{step:tree6}
		\STATE\label{step:tree7} let $V_C\!=\!\{w_1,\!...,w_{|V_C|}\}$ where $w_1\!=\!v$, $i^*\!=\!i\!+\!1\!-\!i(C)$, and $\ell^*\!=\!\ell\!+\!|\{u\}\setminus\{v\}|\!-\!\ell(C)$.
		\STATE\label{step:tree8} let L be a matrix that has an entry $[j,i',\ell']$ for all $1\leq j\leq |V_C|$, $0\leq i'\leq i^*$ and $0\leq \ell'\leq \ell^*$, which is initialized to $\emptyset$.
		\STATE\label{step:tree9} L$[1,\!i',\!\ell']\!\Leftarrow\!\{U\!\in\!\mathrm{M}[v,\!v,\!i',\!\ell']\!: U\!\cap\! V_C\!=\!\emptyset\}$ for all $\displaystyle{0\!\leq\! i'\!\leq\! i^*}$ and $\displaystyle{0\!\leq\! \ell'\!\leq\! \ell^*}$ s.t. $i'\!+\!\ell'\!\leq\! (k\!+\!t)/d$.
	
		\FOR{$j=2,\ldots,|V_C|$, $i'=0,\ldots,i^*$, $\ell'=0,\ldots,\ell^*$}\label{step:tree10}
			\STATE\label{step:tree11} let ${\cal A}$ be the family of all sets $U\!\cup\! W$ such that $U\!\cap\! (W\!\cup\! V_C)\!=\!\emptyset$, and there are $\displaystyle{0\!\leq\! i''\!\leq\! i'}$ and $\displaystyle{0\!\leq\! \ell''\!\leq\! \ell'}$ satisfying $i''\!+\!\ell''\!\leq\! \displaystyle{\frac{k\!+\!t}{d}}$ for which\\(1) $U\!\in\! \mathrm{M}[f_C(w_j),\!w_j,\!i'',\!\ell'']\}$ and $W\!\in\! \mathrm{L}[j\!-\!1,\!i'\!-\!i'',\!\ell'\!-\!\ell'']$; or\\
			(2) $w_j$ is not a leaf in $C$, $\ell''\!\geq\! 1$,$U\!\in\! \mathrm{M}[w_j,\!w_j,\!i'',\!\ell'']\}$ and $W\!\in\! \mathrm{L}[j,\!i'\!-\!i'',\!\ell'\!-\!\ell'']$.
			
			\STATE\label{step:tree12} L$[j,i',\ell']\Leftarrow$ \alg{RepAlg}$(V,k+t,{\cal A})$.	
		\ENDFOR\label{step:tree13}
			
		\STATE\label{step:tree14} N$[C]\Leftarrow \{U\cup(V_C\setminus\{v,u\}): U\in\mathrm{L}[|V_C|,i^*,\ell^*]\}$.
	\ENDFOR

  \STATE\label{step:tree16} M$[v,u,i,\ell]\Leftarrow$ \alg{RepAlg}$(V,k+t,\bigcup_{C\in{\cal C}_{v,u}}\mathrm{N}[C])$.
\ENDFOR

\STATE\label{step:tree18} accept iff $\mathrm{M}[r,r,k-1,t]\neq\emptyset$.
\end{algorithmic}
\end{algorithm}

\myparagraph{Correctness and Running Time} The following lemma implies the correctness of \alg{TreeAlg} (the proof is given in Appendix~\ref{app:lemmakltreecor}).

\begin{mylemma}\label{lemma:kltreecor}
For all $v,u\!\in\! V$, $0\!\leq\! i\!<\! k$ and $0\!\leq\! \ell\!\leq\! t$, M$[v,\!u,\!i,\!\ell]$ represents $Sol_{v,u,i,\ell}$.
\end{mylemma}

{\noindent We summarize in the next result.}
\begin{mylemma}\label{lemma:kltree} \alg{TreeAlg} solves {\sc $(k,t)$-Tree} in time $O(2.61804^{k+t}|V|^{O(1)})$.
\end{mylemma}

\begin{proof}
Let $q\!=\!k\!+\!t$, and $0\!<\!\epsilon\!<\!1$ be some constant. Choose some constant $d\!\geq\! 2$ satisfying ${cq \choose q/d} \!=\! O(2^{\epsilon q})$ and $\frac{1}{d}\!\leq\! \epsilon$.

Lemma \ref{lemma:kltreecor} and Step \ref{step:tree18} imply that \alg{TreeAlg} solves {\sc $(k,\!t)$-Tree}. It is easy to verify that for any $0\!\leq\! r^*\!\leq\! q$ and call \alg{RepAlg}$(V,\!k\!+\!t,\!{\cal S})$ executed by \alg{TreeAlg}, where ${\cal S}$ is a family of $r^*$-subsets of $V$, there is $0\leq r'\leq \min\{r^*,q/d\}$ such that $\displaystyle{|{\cal S}|\!\leq\!2^{o(q)}|V|^{O(d)}(\frac{(cq)^{q}}{(r^*\!-\!r')^{r^*\!-\!r'}(cq\!-\!(r^*\!-\!r'))^{q\!-\!(r^*\!-\!r')}})(\frac{(cq)^q}{{r'}^{r'}(cq\!-\!r')^{q\!-\!r'}})}$. Thus, the way \alg{RepAlg} proceeds implies that \alg{TreeAlg} runs in time

\[\begin{array}{l}
\vspace{1em}

\displaystyle{O(\!2^{o(\!q\!)}\!|V|\!^{O(\!d\!)}\!\max_{r=0}^q\max_{r'=0}^{\min\{q\!-\!r,q/d\}}\!\left\{\!(\frac{(cq)^{q}}{r^r\!(cq\!-\!r)^{q\!-\!r}})\!(\frac{(cq)^q}{{r'}^{r'}\!(cq\!-\!r')^{q\!-\!r'}})\!(\frac{cq}{cq\!-\!(r\!+\!r')})\!^{q\!-\!(r\!+\!r')}\!\right\}\!)}\\

\vspace{1em}





\displaystyle{=\!O(\!2^{o(q)}|V|^{O(1)}\!\max_{r=0}^q\!\left\{\!(\frac{(cq)^{q}}{r^r(cq\!-\!r)^{q\!-\!r}}){cq \choose q/d}(\frac{(c\!+\!1/d)q}{cq\!-\!r})^{q\!-\!r}\!\right\}\!)}\\

\vspace{0.3em}

\displaystyle{=\!O(\!2^{\epsilon q \!+\! o(q)}|V|^{O(1)}\!\max_{r=0}^q\!\left\{\!(\frac{(cq)^{q}}{r^r(cq\!-\!r)^{q\!-\!r}})(\frac{(c+\epsilon)q}{cq\!-\!r})^{q\!-\!r}\!\right\}\!).}
\end{array}
\]
By choosing $c\!=\!1.447$ and a small enough $\epsilon>0$, the maximum is obtained at $r\!=\!\alpha k$, where $\alpha\!\cong\!0.55277$. Thus, \alg{TreeAlg} runs in  time $O(2.61804^{k+t}|V|^{O(1)})$.\qed
\end{proof}
Lemmas \ref{lemma:tokltree} and \ref{lemma:kltree} imply the following theorem.

\begin{theorem}
{\sc $k$-IOB} can be solved in time $O(6.85414^k|V|^{O(1)})$.
\end{theorem}

\mysection{Improving Known Applications}\label{section:knownapp}

Fomin et al. \cite{representative} proved that {\sc Long Directed Cycle}, {\sc Weighted $k$-Path} and {\sc Weighted $k$-Tree} can be solved in times $O(8^k|E|\!\log^2\!|V|)$, $O(2.851^k|V|\!\log^2\!|V|)$ and $O(2.851^k|V|^{O(1)})$, respectively. By replacing their computation of representative families with our scheme, \alg{RepAlg}, we obtain for these problems exact algorithms with improved running times of $O(6.75^k|E|\!\log^2|V|)$, $O(2.619^k|V|\!\log^2|V|)$ and $O(2.619^k|V|^{O(1)})$, respectively. We give the details in Appendix \ref{app:knownapp}.

\bibliographystyle{splncs03}
\bibliography{references}

\newpage

\appendix

\mysection{An Illustration of a Representative Family}\label{app:illustration}

\begin{figure}
\medskip
\centering
\frame{
\includegraphics[scale=0.29]{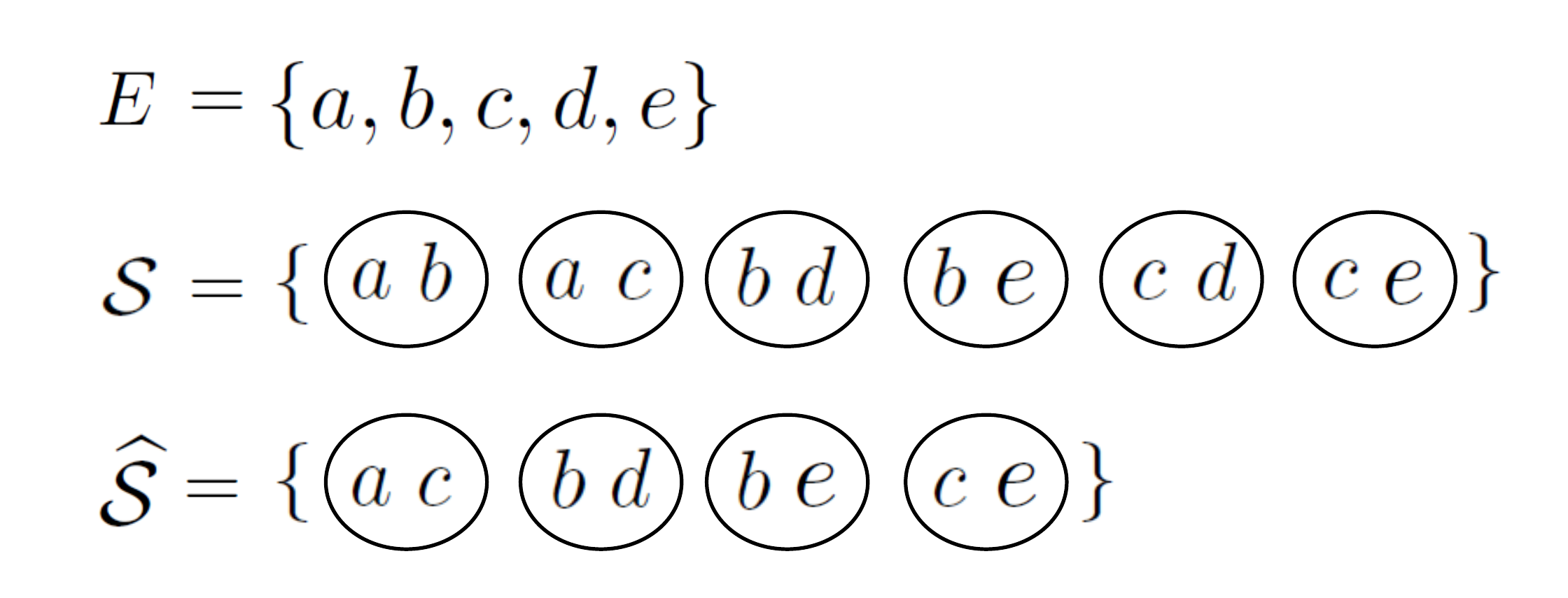}
}
\caption{An illustration of a subfamily $\widehat{\cal S}\!\subseteq\!{\cal S}$ that represents ${\cal S}$ in a matroid $U_{5,4}$, with $p\!=\!2$.}
\label{fig:2}
\smallskip
\end{figure}

{\noindent Figure~\ref{fig:2} illustrates a representative family corresponding to the matroid $U_{5,4}\!=\!\{E,\!{\cal I}\}$, where ${\cal I}\!=\!\{X\!\subseteq\! E\!:\! |X|\!\leq\! 4\}$.  The subfamily $\widehat{\cal S}\!\subseteq\!{\cal S}$ represents ${\cal S}$, since for every pair of sets $X\!\in\!{\cal S}$, and $Y\!\subseteq\!E\setminus X$ such that $|Y|\!\leq\!(k\!-\!p)\!=\!2$, there is a set $\widehat{X}\!\in\!\widehat{\cal S}$ disjoint from $Y$. For example, the set $X\!=\!\{a,\!b\}\in{\cal S}$ is disjoint from $\{c,\!d\},\{c,\!e\}$, and $\{d,\!e\}$. Indeed, the subfamily $\widehat{\cal S}$ contains the set $\{b,\!e\}$ that is disjoint from $\{c,\!d\}$, the set $\{b,\!d\}$ that is disjoint from $\{c,\!e\}$, and the set $\{a,\!c\}$ that is disjoint from $\{d,\!e\}$.
}

\mysection{Proof of Theorem \ref{theorem:newrep} (Continued)}\label{app:theoremcont}

In the rest of the proof, we use the following notation. Given nonnegative integers $p,s$ and $t$, let ${\cal Z}^p_{s,t}$ denote the set of tuples $(p_1,\!p_2,\ldots,p_t)$ of nonnegative integers, each of value at most $s$, whose sum is $p$. Note that $|{\cal Z}^p_{s,t}|\!\leq\! {p+t-1\choose t-1}\!\leq\! 2^{t\!\log (p + t)}$.

Recall that our goal is to decrease the $\tau_{\cal F}^1$ and $\tau_{\chi}^1$ parameters of Separator 1 to be within the required bounds for $\tau_{\cal F}^*$ and $\tau_{\chi}^*$. We achieve this goal by repeatedly applying the following two lemmas:

\begin{mylemma}[\cite{representative}]\label{lemma:const2}
If there is an $(\!n,\!k,\!p\!)$-separator with parameters ${\cal C}(\!n,\!k,\!p),\tau_{\cal F}(\!n,$\\$k,\!p),\Delta(n,\!k,\!p)$ and $\tau_{\chi}(n,\!k,\!p)$, then there is an $(n,\!k,\!p)$-separator with parameters:
\begin{itemize}
\item $\displaystyle{{\cal C}'(n,k,p)\leq {\cal C}(k^2,k,p)k^{O(1)}\log n}$.
\item $\displaystyle{\tau_{\cal F}'(n,k,p) = O(\tau_{\cal F}(k^2,k,p) + {\cal C}(k^2,k,p)k^{O(1)}n\log n)}$.
\item $\displaystyle{\Delta'(n,k,p)\leq \Delta(k^2,k,p)k^{O(1)}\log n}$.
\item $\displaystyle{\tau_{\chi}'(n,k,p) = O(\tau_{\chi}(k^2,k,p) + \Delta(k^2,k,p)k^{O(1)}\log n)}$.
\end{itemize}\end{mylemma}

\begin{mylemma}[\cite{representative}]\label{lemma:const3}
Let $s\!=\!\lfloor\log^2k\rfloor$ and $t\!=\!\lceil\frac{k}{s}\rceil$. If there is an $(n,\!k,\!p)$-separator with parameters ${\cal C}(n,\!k,\!p),\tau_{\cal F}(n,\!k,\!p),\Delta(n,\!k,\!p)$ and $\tau_{\chi}(n,\!k,\!p)$, then there is an $(n,\!k,\!p)$-separator with parameters:
\begin{itemize}
\item $\displaystyle{{\cal C}'(n,k,p)\leq 2^{O(t\log n)}\cdot\sum_{(p_1,\ldots,p_t)\in{\cal Z}^p_{s,t}}\prod_{i\leq t}{\cal C}(n,s,p_i)}$.
\item $\displaystyle{\tau_{\cal F}'(n,k,p) = O(\sum_{\widehat{p}\leq s}\tau_{\cal F}(n,s,\widehat{p}) + {\cal C}'(n,k,p)n^{O(1)})}$.
\item $\displaystyle{\Delta'(n,k,p)\leq 2^{O(t\log n)}\cdot\max_{(p_1,\ldots,p_t)\in{\cal Z}^p_{s,t}}\prod_{i\leq t}\Delta(n,s,p_i)}$.
\item $\displaystyle{\tau_{\chi}'(n,k,p) = O(\Delta'(n,k,p)\cdot n^{O(1)} + 2^{O(t\log n)}\cdot\sum_{\widehat{p}\leq s}\tau_{\chi}(n,s,\widehat{p}))}$.
\end{itemize}
\end{mylemma}
Applying Lemma \ref{lemma:const2} to Separator 1, we get a separator with parameters:

\begin{itemize}
\item $\displaystyle{{\cal C}^2(n,k,p) = O(\frac{(ck)^{k}}{p^p(ck\!-\!p)^{k\!-\!p}}k^{O(1)}\!\log n)\! =\! O({ck\choose p}(\frac{ck\!-\!p}{ck})^{(c\!-\!1)k}k^{O(1)}\!\log n)}$.
\item $\displaystyle{\tau_{\cal F}^2(n,k,p) = 2^{k^{O(k)}} + \frac{(ck)^{k}}{p^p(ck-p)^{k-p}}k^{O(1)}n\log n}$.
\item $\displaystyle{\Delta^2(n,k,p) = O((\frac{ck}{ck-p})^{k-p}k^{O(1)}\log n)}$.
\item $\displaystyle{\tau_{\chi}^2(n,k,p) = k^{O(k)}\log n}$.
\end{itemize}
We now apply Lemma \ref{lemma:const3} to this separator. Recall that in Lemma \ref{lemma:const3}, we set $s=\lfloor\log^2k\rfloor$ and $t=\lceil\frac{k}{s}\rceil$. This yields a separator with parameters:

\[\begin{array}{ll}
{\cal C}^3(n,k,p) & \leq \displaystyle{2^{O(t\log n)}\cdot\sum_{(p_1,\ldots,p_t)\in{\cal Z}^p_{s,t}}\prod_{i\leq t}{\cal C}^2(n,s,p_i)}\\

&\leq \displaystyle{2^{O(t\log n)}\max_{(p_1,\ldots,p_t)\in{\cal Z}^p_{s,t}}\prod_{i\leq t}[{cs\choose p_i}(\frac{cs-p_i}{cs})^{(c-1)s}s^{O(1)}\log n]}\\

&\leq \displaystyle{2^{O(t\log n)}\cdot{ck\choose p}\cdot\max_{(p_1,\ldots,p_t)\in{\cal Z}^p_{s,t}}\prod_{i\leq t}(1-\frac{p_i}{cs})^{(c-1)s}}\\

&\leq \displaystyle{2^{O(t\log n)}\cdot {ck\choose p}(1-\frac{p}{ck})^{(c-1)k}}
\end{array}\]


\[\begin{array}{lll}
\tau_{\cal F}^3(n,\!k,\!p)\!&=\!&\displaystyle{O(\sum_{\widehat{p}\leq s}\tau_{\cal F}^2(n,s,\widehat{p}) + {\cal C}^3(n,k,p)n^{O(1)})}\\

&\leq\!&\displaystyle{2^{s^{O(s)}}\!+\!\max_{\widehat{p}\leq s}(\frac{(cs)^s}{\widehat{p}^{\widehat{p}}(cs\!-\!\widehat{p})^{s\!-\!\widehat{p}}})s^{O(1)}n\!\log n + 2^{O(t\!\log n)}{ck\choose p}(1\!-\!\frac{p}{ck})^{(c\!-\!1)k}}\\

&\leq\!&\displaystyle{2^{(\log k)^{O(\log^2 k)}} + 2^{O(t\log n)}\cdot {ck\choose p}(1-\frac{p}{ck})^{(c-1)k}}
\end{array}\]


\[\begin{array}{lll}
\Delta^3(n,k,p) &\leq \displaystyle{2^{O(t\log n)}\cdot\max_{(p_1,\ldots,p_t)\in{\cal Z}^p_{s,t}}\prod_{i\leq t}\Delta^2(n,s,p_i)}&\\

&\leq \displaystyle{2^{O(t\log n)}\cdot \max_{(p_1,\ldots,p_t)\in{\cal Z}^p_{s,t}}\prod_{i\leq t}[(\frac{cs}{cs-p_i})^{s-p_i}s^{O(1)}\log n]}&\\

&\leq \displaystyle{2^{O(t\log n)}\cdot\max_{(q_1,\ldots,q_t)\in{\cal Z}^{k-p}_{s,t}}\prod_{i\leq t}(\frac{cs}{(c-1)s+q_i})^{q_i}}&\\

&\leq \displaystyle{2^{O(t\log n)}\cdot(ck)^{k-p}\cdot\max_{(q_1,\ldots,q_t)\in{\cal Z}^{k-p}_{s,t}}\prod_{i\leq t}(\frac{1}{((c-1)s+q_i)t})^{q_i}}& = (*)
\end{array}\]
By Gibbs' inequality, we have that

\[(*)\leq \displaystyle{2^{O(t\log n)}\cdot(ck)^{k-p}\cdot\frac{1}{((c-1)st+(k-p))^{k-p}}}\leq \displaystyle{2^{O(t\log n)}\cdot(\frac{ck}{ck-p})^{k-p}}\]



\[\begin{array}{ll}
\tau_{\chi}^3(n,k,p) &= \displaystyle{O(2^{O(t\log n)}(\frac{ck}{ck-p})^{k-p}n^{O(1)} + 2^{O(t\log n)}\sum_{\widehat{p}\leq s}\tau_{\chi}^2(n,s,\widehat{p}))}\\

&\leq \displaystyle{2^{O(t\log n)}\cdot(\frac{ck}{ck-p})^{k-p} + 2^{O(t\log n)}s^{O(s)}\log n}\\

&\leq \displaystyle{2^{O(t\log n)}\cdot(\frac{ck}{ck-p})^{k-p}}
\end{array}\]
Applying Lemma \ref{lemma:const2} to this separator, we get a separator with parameters:

\begin{itemize}
\item $\displaystyle{{\cal C}^4(n,k,p)\leq 2^{O(\frac{k}{\log k})}\cdot {ck\choose p}(1-\frac{p}{ck})^{(c-1)k}\log n}$.

\item $\displaystyle{\tau_{\cal F}^4(n,k,p) = O(2^{(\log k)^{O(\log^2k)}} + 2^{O(\frac{k}{\log k})}\cdot {ck\choose p}(1-\frac{p}{ck})^{(c-1)k}n\log n)}$.

\item $\displaystyle{\Delta^4(n,k,p)\leq 2^{O(\frac{k}{\log k})}\cdot(\frac{ck}{ck-p})^{k-p}\log n}$.

\item $\displaystyle{\tau_{\chi}^4(n,k,p) = O(2^{O(\frac{k}{\log k})}\cdot(\frac{ck}{ck-p})^{k-p}\log n)}$.
\end{itemize}
We next apply Lemma \ref{lemma:const3} again. As in the analysis of the third separator, we set $s=\lfloor\log^2k\rfloor$ and $t=\lceil\frac{k}{s}\rceil$. This yields a separator with parameters:

\[\begin{array}{ll}
{\cal C}^5(n,k,p) &\leq \displaystyle{2^{O(t\log n)}\cdot\sum_{(p_1,\ldots,p_t)\in{\cal Z}^p_{s,t}}\prod_{i\leq t}{\cal C}^4(n,s,p_i)}\\

&\leq \displaystyle{2^{O(t\log n)}\max_{(p_1,\ldots,p_t)\in{\cal Z}^p_{s,t}}\prod_{i\leq t}[2^{O(\frac{s}{\log s})}\cdot {cs\choose p_i}(1-\frac{p_i}{cs})^{(c-1)s}\log n]}\\

&\leq \displaystyle{2^{O(t\log n + \frac{k}{\log\log k})}{ck\choose p}\max_{(p_1,\ldots,p_t)\in{\cal Z}^p_{s,t}}[\prod_{i\leq t}(1-\frac{p_i}{cs})]^{(c-1)s}}\\

&\leq \displaystyle{2^{O(t\log n + \frac{k}{\log\log k})}{ck\choose p}(1-\frac{p}{ck})^{(c-1)k}}\\

&\leq \displaystyle{2^{O(t\log n + \frac{k}{\log\log k})}\frac{(ck)^{k}}{p^p(ck-p)^{k-p}}}
\end{array}\]


\[\begin{array}{lll}
\tau_{\cal F}^5(n,k,p) &=& \displaystyle{O(\sum_{\widehat{p}\leq s}\tau_{\cal F}^4(n,s,\widehat{p}) + {\cal C}^5(n,k,p)n^{O(1)})}\\

&\leq & \displaystyle{s2^{(\log s)^{O(\log^2 s)}} + \max_{\widehat{p}\leq s}(2^{O(\frac{s}{\log s})}\cdot {cs\choose \widehat{p}}(1-\frac{\widehat{p}}{cs})^{(c-1)s}n\log n)}\\

&&+ \displaystyle{2^{O(t\log n + \frac{k}{\log\log k})}\frac{(ck)^{k}}{p^p(ck-p)^{k-p}}}\\

&\leq & \displaystyle{(\log^2 k)2^{(\log\log k)^{O(\log^2\log k)}} + 2^{O(t\log n + \frac{k}{\log\log k})}\frac{(ck)^{k}}{p^p(ck-p)^{k-p}}}\\

&\leq & \displaystyle{2^{O(t\log n + \frac{k}{\log\log k})}\frac{(ck)^{k}}{p^p(ck-p)^{k-p}}}

\end{array}\]
In the following analysis of $\Delta^5(n,k,p)$, the last transition is achieved as in the analysis of $\Delta^3(n,k,p)$.

\[\begin{array}{ll}
\Delta^5(n,k,p) &\leq \displaystyle{2^{O(t\log n)}\cdot\max_{(p_1,\ldots,p_t)\in{\cal Z}^p_{s,t}}\prod_{i\leq t}\Delta^4(n,s,p_i)}\\

&\leq \displaystyle{2^{O(t\log n)}\max_{(p_1,\ldots,p_t)\in{\cal Z}^p_{s,t}}\prod_{i\leq t}(2^{O(\frac{s}{\log s})}\cdot(\frac{cs}{cs-p_i})^{s-p_i}\log n)}\\

&\leq \displaystyle{2^{O(t\log n + \frac{k}{\log\log k})}\max_{(p_1,\ldots,p_t)\in{\cal Z}^p_{s,t}}\prod_{i\leq t}(\frac{cs}{cs-p_i})^{s-p_i}}\\

&\leq \displaystyle{2^{O(t\log n + \frac{k}{\log\log k})}(\frac{ck}{ck-p})^{k-p}}
\end{array}\]


\[\begin{array}{ll}
\tau_{\chi}^5(n,k,p) &= \displaystyle{2^{O(t\log n + \frac{k}{\log\log k})}(\frac{ck}{ck-p})^{k-p}n^{O(1)} + 2^{O(t\log n)}\sum_{\widehat{p}\leq s}\tau_{\chi}^4(n,s,\widehat{p})}\\

&\leq \displaystyle{2^{O(t\log n + \frac{k}{\log\log k})}(\frac{ck}{ck-p})^{k-p} + 2^{O(t\log n)}s^{O(s)}\log n}\\

&\leq \displaystyle{2^{O(t\log n + \frac{k}{\log\log k})}(\frac{ck}{ck-p})^{k-p}}

\end{array}\]
Applying Lemma \ref{lemma:const2} to this separator, we get a separator with parameters:

\begin{itemize}
\item $\displaystyle{{\cal C}^6(n,k,p)\leq 2^{O(\frac{k}{\log\log k})}\frac{(ck)^{k}}{p^p(ck-p)^{k-p}}\log n}$.
\item $\displaystyle{\tau_{\cal F}^6(n,k,p)\leq 2^{O(\frac{k}{\log\log k})}\frac{(ck)^{k}}{p^p(ck-p)^{k-p}}n\log n}$.
\item $\displaystyle{\tau_{\chi}^6(n,k,p)\leq 2^{O(\frac{k}{\log\log k})}(\frac{ck}{ck-p})^{k-p}\log n}$.
\end{itemize}
This separator has the desired parameters ${\cal C}^*(n,k,p)$, $\tau_{\cal F}^*(n,k,p)$ and $\tau_{\chi}^*(n,k,p)$, and we thus conclude the proof of the theorem.\qed

\vspace{-0.25em}

\mysection{Pseudocode for the Computation Scheme \alg{RepAlg}}\label{app:scheme}

\vspace{-0.4em}

We give below the pseudocode for \alg{RepAlg}, the computation scheme developed in Section \ref{section:computation}.

\renewcommand{\thealgorithm}{3}
\floatname{algorithm}{Algorithm}
\begin{algorithm}[!ht]
\caption{\alg{RepAlg}($E,k,{\cal S},w$)}
\begin{algorithmic}[1]\label{alg:scheme}
\STATE let $n'=\log^4\log^2 k$, and $k'=\log^2\log^2 k$.

\FOR{$p'=0,\ldots,\min\{k',p\}$}\label{step:repalg1}
	\STATE construct an $(n',\!k',\!p')$-separator with parameters ${\cal C}^1(n',\!k',\!p')$, $\tau_{\cal F}^1(n',\!k',\!p')$, $\Delta^1(n',\!k',\!p')$, and $\tau_{\chi}^1(n',\!k',\!p')$.
\ENDFOR

\STATE\label{step:repalg2} given the results of Step \ref{step:repalg1}, repeatedly apply Lemmas \ref{lemma:const2} and \ref{lemma:const3} to construct an $(n,p,k)$-separator with parameters ${\cal C}^*(n,k,p)$, $\tau_{\cal F}^*(n,k,p)$, and $\tau_{\chi}^*(n,k,p)$.

\STATE given the result of Step \ref{step:repalg2}, compute the corresponding family ${\cal F}$ and subfamilies $\chi(S)$, for all $S\!\in\!{\cal S}$.

\FORALL{$F\!\in\!{\cal F}$}
	\STATE let $z_F\!\in\!\{0,\!1\}$ be an indicator variable for using $F$, initialized to 0.
\ENDFOR
	
\STATE order ${\cal S}\!=\!\{S_1,\ldots,S_{|{\cal S}|}\}$, such that $w(S_{i-1})\!\geq\! w(S_i)$ ($w(S_{i-1})\!\leq\! w(S_i)$), for all $2\!\leq\! i\!\leq\! |{\cal S}|$.

\STATE initialize $\widehat{\cal S}\Leftarrow\emptyset$.

\FOR{$i=1,\ldots,|{\cal S}|$}

	\STATE {\bf for all} $F\!\in\!\chi(S_i)$ such that $z_F\!=\!0$ {\bf do} $\widehat{\cal S}\Leftarrow \widehat{\cal S}\cup\{S_i\}$, and $z_F\Leftarrow 1$. {\bf end for}

\ENDFOR

\STATE return $\widehat{\cal S}$.
\end{algorithmic}
\end{algorithm}

\vspace{-0.25em}

\mysection{Proof of Lemma \ref{lemma:coveradd}}\label{app:lemmacoveradd}

\vspace{-0.4em}

In this section we prove the following lemma.

\vspace{-0.25em}

\begin{mylemma}\label{lemma:generalcov}
Consider an iteration of Step \ref{pcalg:ite}	in \alg{PCAlg}, corresponding to some~values $i,j$ and $\ell$. Let ${\cal A}\!=\!(\bigcup_{1\leq j'\leq j}\mathrm{M}[i\!-\!1,j',\ell\!-\!1])\cup\{\emptyset\}$. Then, for all $0\!\leq\! r'\!\leq\! r$ and $0\!\leq\! j'\!\leq\! j$, ${\cal A}_{r',j'}$ represents ${\cal A}^*_{r',j'} \!=\! \{S\!\cup\! S'_i: S\!\in\! {\cal A}, S'_i\!\subseteq \!\{s_1,\!...,s_{r'}\}, |S\!\cup\!S'_i|\!=\!j'\}$.
\end{mylemma}

\vspace{-0.25em}

{\noindent Note that Lemma \ref{lemma:coveradd} is a special case of Lemma \ref{lemma:generalcov} when $r'\!=\!r$, and $j'\!=\!j$.}

\vspace{-0.25em}

\begin{proof} By Step \ref{pcalg:new1}, the claim holds for $r'\!=\!0$ and all $0\!\leq\! j'\!\leq\! j$. Next consider some $1\!\leq\! r'\!\leq\! r$, and assume that the claim holds for all $0\!\leq\! r''\!<\!r'$ and $0\!\leq\! j''\!\leq\! j$. By Step \ref{pcalg:new3} and Observation \ref{obs:transitive}, it is enough to prove that ${\cal B} \!=\! \{S\!\cup\!\{\!s_{r'}\!\}\!:\! j'\!\geq\! 1, S\!\in\! {\cal A}_{r'\!-\!1,j'\!-\!1}, s_{r'}\!\notin\! S\}\!\cup\!{\cal A}_{r'\!-\!1,j'}$ represents ${\cal A}^*_{r',j'}$. First, we get that ${\cal B}\!\subseteq\! {\cal A}^*_{r',j'}$ as~follows:

\vspace{-0.7em}

\[
\begin{array}{ll}

{\cal B} & \displaystyle{= \{S\!\cup\!\{s_{r'}\}\!:\! j'\!\geq\! 1, S\!\in\! {\cal A}_{r'\!-\!1,j'\!-\!1}, s_{r'}\!\notin\! S\}\!\cup\!{\cal A}_{r'\!-\!1,j'}}\\

& \displaystyle{\subseteq^{(1)} \{S\!\cup\!\{s_{r'}\}\!:\! j'\!\geq\! 1, S\!\in\! {\cal A}^*_{r'\!-\!1,j'\!-\!1}, s_{r'}\!\notin\! S\}\!\cup\!{\cal A}^*_{r'\!-\!1,j'}}\\

& \displaystyle{=^{(2)} \{S\!\cup\! S'\!:\! S\!\in\! {\cal A}, S'\!\subseteq \!\{s_1,\ldots,s_{r'}\}, |S\!\cup\!S'|\!=\!j'\}} = \displaystyle{{\cal A}^*_{r',j'}}.
\end{array}
\]

\vspace{-0.5em}

$(1)$ by the induction hypothesis.

$(2)$ by the definitions of ${\cal A}^*_{r'-1,j'-1}$,${\cal A}^*_{r'-1,j'}$.

\medskip

{\noindent Now, consider some $X\!\in\! {\cal A}^*_{r',j'}$ and $Y\!\subseteq\! (U\!\setminus\! X)$ such that $|Y|\!\leq\! k\!-\!j'$. Since $X\!\in\! {\cal A}^*_{r',j'}$, either $(X\setminus\{s_{r'}\})\!\in\! {\cal A}^*_{r'\!-\!1,j'\!-\!1}$ or $X\!\in\! {\cal A}^*_{r'\!-\!1,j'}$. In the first case, by the induction hypothesis, there is $\widehat{X}\!\in\! {\cal A}_{r'\!-\!1,j'\!-\!1}$ such that $\widehat{X}\!\cap\!(Y\!\cup\!\{s_{r'}\})\!=\!\emptyset$. Then, $(\widehat{X}\!\cup\!\{s_{r'}\})\!\in\!{\cal A}_{r',j'}$, and $(\widehat{X}\!\cup\!\{s_{r'}\})\cap Y\!=\!\emptyset$. In the second case, by the induction hypothesis, there is $\widehat{X}\!\in\! {\cal A}_{r'\!-\!1,j'}\!\subseteq\!{\cal A}_{r',j'}$ such that $\widehat{X}\!\cap\! Y\!=\!\emptyset$. Thus, ${\cal B}$ indeed represents ${\cal A}^*_{r',j'}$.}\qed
\end{proof}


\vspace{-0.25em}

\mysection{Proof of Lemma \ref{lemma:coveralg}}\label{app:coveralg}

\vspace{-0.4em}

By Step \ref{pcalg:2}, the claim holds for all $1\!\leq\! j\!\leq\! k$ and $0\!\leq\! \ell\!\leq\! k$ when $i\!=\!0$, and for all $0\!\leq\!i\!\leq\! m$ and $1\!\leq\! j\!\leq\! k$ when $\ell\!=\!0$, or $i\!<\!\ell\!\leq\! k$. Next, consider some $1\!\leq\! i\!\leq\! m, 1\!\leq\! j\!\leq\! k$ and $1\!\leq\! \ell\!\leq\! \min\{i,\!k\}$, and assume that the claim holds for $i\!-\!1$, and all $1\!\leq\! j'\!\leq\! j$ and $0\!\leq\! \ell'\!\leq\! \ell$.

Let $Sol^*_{i,j,\ell}$ be the family of partial solutions including $j$ elements, obtained from a subfamily ${\cal S}'$ that includes $S_i$ and $(\ell\!-\!1)$ sets in $\{S_1,\!...,S_{i-1}\}$, i.e., $Sol^*_{i,j,\ell}\! = \!\{S\!\subseteq\!(\bigcup{\cal S'})\!: \{S_i\}\!\subseteq\!{\cal S'}\!\subseteq\!\{S_1,\!...,S_i\}, |S|\!=\!j, |{\cal S'}|\!=\!\ell\}$. We now show that ${\cal A}^*_{i,j,\ell}$ represents $Sol^*_{i,j,\ell}$, which, by Lemma \ref{lemma:coveradd} and Observation \ref{obs:transitive}, implies that ${\cal A}_{r,j}$ represents $Sol^*_{i,j,\ell}$. First, the induction hypothesis implies that $(\bigcup_{1\leq j'\leq j}\mathrm{M}[i\!-\!1,\!j',\!\ell\!-\!1])\!\subseteq\!(\bigcup_{1\leq j'\leq j}Sol_{i\!-\!1,j',\ell\!-\!1})$, and thus ${\cal A}^*_{i,j,\ell} \!\subseteq\! Sol^*_{i,j,\ell}$. Second, consider a pair of sets $X\!\in\! Sol^*_{i,j,\ell}$, and $Y\!\subseteq\! (U\!\setminus\! X)$ such that $|Y|\!\leq\! k\!-\!j$. Since $X\!\in\! Sol^*_{i,j,\ell}$, there are $X'\!\in\! (\bigcup_{j'=1}^jSol_{i\!-\!1,j',\ell\!-\!1})\!\cup\!\{\emptyset\}$, and $S\!\subseteq\! S_i$ such that $X\!=\!X'\!\cup\! S$. If $X'\!=\!\emptyset$, then $X\!\in\!{\cal A}^*_{i,j,\ell}$. Else, by the induction hypothesis, there is $\widehat{X}'\!\in\!\mathrm{M}[i\!-\!1,\!|X'|,\!\ell\!-\!1]$ such that $\widehat{X}'\!\cap\! (Y\!\cup\!(S\!\setminus\! X')) \!=\! \emptyset$. We get that $(\widehat{X}'\!\cup\! (S\!\setminus\! X'))\!\in\! {\cal A}^*_{i,j,\ell}$, and $(\widehat{X}'\!\cup\! (S\!\setminus\! X'))\!\cap\! Y\!=\!\emptyset$. Thus, ${\cal A}^*_{i,j,\ell}$ represents $Sol^*_{i,j,\ell}$.

By the induction hypothesis, $\mathrm{M}[i\!-\!1,\!j,\!\ell]\!\cup\!{\cal A}_{r,j}\!\subseteq Sol_{i\!-\!1,j,\ell}\!\cup\! Sol^*_{i,j,\ell}\!\subseteq Sol_{i,j,\ell}$. Now, consider a pair of sets $X\!\in\! Sol_{i,j,\ell}$, and $Y\!\subseteq\! (U\!\setminus\! X)$ such that $|Y|\!\leq\! k\!-\!j$. Since $X\!\in\! Sol_{i,j,\ell}$, either $X\in Sol_{i\!-\!1,j,\ell}$ or $X\in Sol^*_{i,j,\ell}$. In the first case, the induction hypothesis implies that there is $\widehat{X}\!\in\! \mathrm{M}[i\!-\!1,\!j,\!\ell]$ such that $\widehat{X}\!\cap\! Y\!=\!\emptyset$. In the second case, since ${\cal A}_{r,j}$ represents $Sol^*_{i,j,\ell}$, there is $\widehat{X}\!\in\!{\cal A}_{r,j}$ such that $\widehat{X}\!\cap\! Y\!=\!\emptyset$. We get that $\mathrm{M}[i\!-\!1,\!j,\!\ell]\!\cup\!{\cal A}_{r,j}$ represents $Sol_{i,j,\ell}$. Thus, by Observation \ref{obs:transitive}, M$[i,\!j,\!\ell]$ represents $Sol_{i,j,\ell}$.\qed

\mysection{Proof of Lemma \ref{lemma:kltreecor}}\label{app:lemmakltreecor}

We use in the proof below the next result, implicitly given in \cite{representative}.

\begin{mylemma}\label{lemma:treediv}
For any out-tree $T$ of $G$ rooted at $v$, containing at least three nodes, in which $v\!=\!u$ or $u$ is a leaf, there exists $C\!\in\!{\cal C}_{v,u}$ with whom $T$ complies.
\end{mylemma}

{\noindent\bf Proof of Lemma \ref{lemma:kltreecor}.}
By Steps \ref{step:tree1} and \ref{step:tree2}, the lemma holds for all $v,\!u\!\in\! V$, and $i\!=\!\ell\!=\!0$. Now, consider some $v,\!u\!\in\! V$, $0\!\leq\! i\!<\! k$, and $0\!\leq\! \ell\!\leq\! t$ such that $1\!\leq\! i\!+\!\ell$, and assume that the lemma holds for all $v',\!u'\!\in\! V$, $0\!\leq i'\!\leq\! i$, and $0\!\leq\! \ell'\!\leq\! \ell$ such that $i'\!+\!\ell'\!<\!i\!+\!\ell$.

For all $C\!\in\!{\cal C}_{v,u}$, let ${\cal T}_{v,u,i,\ell,C}$ be the set of out-trees in ${\cal T}_{v,u,i,\ell}$ that comply with $C$. Also, let $Sol_{v,u,i,\ell,C} \!=\! \{V_T\setminus\{v,u\}:T\!\in\!{\cal T}_{v,u,i,\ell,C}\}$. In order to prove the inductive claim, we need the following claim.
\begin{cla}\label{claim:kl1}
For all $C\!\in\! {\cal C}_{v,u}$, N$[C]$ represents $Sol_{v,u,i,\ell,C}$.
\end{cla}
We first show that Claim \ref{claim:kl1} implies the correctness of the inductive claim. By Observation \ref{obs:transitive}, it is enough to prove that ${\cal B} = \bigcup_{C\in{\cal C}_{v,u}}\mathrm{N}[C]$ represents $Sol_{v,u,i,\ell}$. By Claim \ref{claim:kl1}, ${\cal B}\!\subseteq\! \bigcup_{C\in{\cal C}_{v,u}}Sol_{v,u,i,\ell,C}\!\subseteq\! Sol_{v,u,i,\ell}$. Consider some $X\!\in\! Sol_{v,u,i,\ell}$, and $Y\!\subseteq\! (V\setminus X)$ such that $|Y|\!\leq\! k\!+\!t\!-\!(i\!+\!\ell)$. Since $X\!\in\! Sol_{v,u,i,\ell}$, there is an out-tree $T\!\in\! {\cal T}_{v,u,i,\ell}$ whose node set is $X\!\cup\!\{v,u\}$. By Lemma \ref{lemma:treediv}, there is $C\!\in\! {\cal C}_{v,u}$ such that $T\!\in\! {\cal T}_{v,u,i,\ell,C}$. We get that $X\!\in\! Sol_{v,u,i,\ell,C}$. By Claim \ref{claim:kl1}, there is $\widehat{X}\!\in\!\mathrm{N}[C]\!\subseteq\! {\cal B}$ such that $\widehat{X}\!\cap\! Y\!=\!\emptyset$. Thus, ${\cal B}$ represents $Sol_{v,u,i,\ell}$.

We now turn to the proof of Claim~\ref{claim:kl1}. Consider an iteration of Step \ref{step:tree6} that corresponds to some rooted tree $C\!\in\! {\cal C}_{v,u}$. We first note that, formally, a subforest $F'$ of $G$ complies with a subforest $F$ of $C$ if: (1) they have the same roots, (2) $\forall v',\!u'\!\in\!V_F$, $v'$ is an ancestor of $u'$ in $F$ iff $v'$ is an ancestor of $u'$ in $F'$, (3) the leaves in $C$ from $V_F$ are leaves in $F'$, and (4) in the forest obtained by removing $V_F$ from $F'$, each tree has at most $\frac{k\!+\!t}{d}$ nodes and at most two neighbors in $F'$ from $V_F$. Recall that the set ${\cal F}_{v,u,C,j,i',\ell'}$ includes each subforest $F'$ of $G$ that complies with the subforest $F$ of $C$ induced by $\{w_1,\ldots,w_j\}$, such that: (1) $V_{F'}\!\cap\!(V_C\setminus V_F)\!=\!\emptyset$, and (2) the number of internal nodes (leaves) in $V_{F'}$, excluding nodes in $V_F$, is $i'$ ($\ell'$). We denote $Sol_{v,u,C,j,i',\ell'} \!=\! \{V_F\setminus V_C: F\!\in\! {\cal F}_{v,u,C,j,i',\ell'}\}$. In order to prove Claim \ref{claim:kl1}, we need the following claim.

\begin{cla}\label{claim:kl2}
For all $1\!\leq\! j\!\leq\! |V_C|$, $0\!\leq\! i'\!\leq\! i^*$ and $0\!\leq\! \ell'\!\leq\! \ell^*$, we have that $\mathrm{L}[j,\!i',\!\ell']\!\subseteq\!Sol_{v,u,C,j,i',\ell'}$ and $\{U\!\cup\!(V_C\setminus\{v,\!u\})\!: U\!\in\!\mathrm{L}[j,\!i',\!\ell']\}$ represents $\{U\!\cup\!(V_C\setminus\{v,\!u\})\!: U\!\in\!Sol_{v,u,C,j,i',\ell'}\}$.
\end{cla}
We first show that Claim \ref{claim:kl2} implies the correctness of Claim \ref{claim:kl1}. By Claim \ref{claim:kl2},  N$[C] \!=\!  \{U\!\cup\!(V_C\setminus\{v,\!u\})\!: U\!\in\!\mathrm{L}[|V_C|,\!i^*,\!\ell^*]\}\!\subseteq\! \{U\!\cup\! (V_C\setminus\{v,\!u\})\!: U\!\in\! Sol_{v,u,C,|V_C|,i^*,\ell^*}\} \!=\! Sol_{v,u,i,\ell,C}$. Consider some $X\!\in\! Sol_{v,u,i,\ell,C}$, and $Y\!\subseteq\! (V\setminus X)$ such that $|Y|\!\leq\! k\!+\!t\!-\!(i\!+\!\ell)$. Since $X\!\in\! Sol_{v,u,i,\ell,C}$, $U\!=\!X\setminus V_C\!\in\! Sol_{v,u,C,|V_C|,i^*,\ell^*}$.
By Claim \ref{claim:kl2}, there is $\widehat{U}\!\in\!\mathrm{L}[|V_C|,\!i^*,\!\ell^*]$ such that $\widehat{U}\!\cap\! Y\!=\!\emptyset$. We get that $\widehat{X}\!=\!\widehat{U}\!\cup\!(V_C\setminus \{v,\!u\})\in\mathrm{N}[C]$ and $\widehat{X}\!\cap\! Y\!=\!\emptyset$. Thus, Claim \ref{claim:kl1} is correct.

Finally, we turn to the proof of Claim \ref{claim:kl2}, which by the above arguments, concludes the correctness of the lemma. By Steps \ref{step:tree8} and \ref{step:tree9} and the induction hypothesis for M, the claim holds for all $0\!\leq\! i'\!\leq\! i^*$ and $0\!\leq\! \ell'\!\leq\! \ell^*$ when $j=1$. Now, consider some $0\!\leq\! i'\!\leq\! i^*$, $0\!\leq\! \ell'\!\leq\! \ell^*$ and $2\!\leq\! j\!\leq\! |V_C|$, and assume that the claim holds for all $0\!\leq\! i''\!\leq\! i'$, $0\!\leq\! \ell''\!\leq\! \ell'$ and $1\!\leq\! j'\!\leq\! j$ s.t. $(j'\!<\!j$ or $i''\!+\!\ell''\!<\! i'\!+\!\ell')$. By observation \ref{obs:transitive}, it is enough to prove that $\{U\!\cup\!(V_C\setminus\{v,\!u\})\!:U\!\in\!{\cal A}\}$ represents $\{U\!\cup\!(V_C\setminus\{v,\!u\})\!: U\!\in\!Sol_{v,u,C,j,i',\ell'}\}$.

Note that a set $X$ belongs to $Sol_{v,u,C,j,i',\ell'}$ iff there are sets $U,\!W\!\subseteq\! X$, $0\!\leq\! i''\!\leq\!i'$ and $0\!\leq\! \ell''\!\leq\!\ell'$ satisfying $\displaystyle{i''\!+\!\ell''\!\leq\! \frac{k+t}{d}}$, such that $X\!=\!U\!\cup\! W$, $U\!\cap\! (W\!\cup\! V_C)\!=\!\emptyset$ and at least one of the following conditions holds.
\begin{enumerate}
\item $U\!\in\! Sol_{f(w_j),w_j,i'',\ell''}$ and $W\!\in\! Sol_{v,u,C,j-1,i'-i'',\ell'-\ell''}$.
\item $w_j$ is not a leaf in $C$, $\ell''\!\geq\!1$, $U\!\in\! Sol_{w_j,w_j,i'',\ell''}$ and $W\!\in\! Sol_{v,u,C,j,i'-i'',\ell'-\ell''}$.
\end{enumerate}
Thus, by Step \ref{step:tree11} and the inductive hypotheses for M and L, we get that ${\cal A}\!\subseteq\! Sol_{v,u,C,j,i',\ell'}$. Consider some $X\!\in\! Sol_{v,u,C,j,i',\ell'}$, and $Y\!\subseteq\! V\setminus (X\!\cup\! (V_C\setminus\{v,\!u\}))$ such that $|Y|\!\leq\! k\!+\!t-\!(i'\!+\!\ell'+|V_C\setminus\{v,\!u\}|)$. Since $X\!\in\! Sol_{v,u,C,j,i',\ell'}$, there are $U$, $W$, $i''$ and $j''$ as mentioned above. By the inductive hypotheses for M and L, there are sets $\widehat{U}$ and $\widehat{W}$ such that $\widehat{U}\!\cap\!(\widehat{W}\!\cup\! V_C)\!=\!\emptyset$ and $(\widehat{U}\!\cup\!\widehat{W})\!\cap\!Y\!=\!\emptyset$, for which at least one of the following conditions holds.

\begin{enumerate}
\item $\widehat{U}\!\in\!\mathrm{M}[f(w_j),\!w_j,\!i'',\!\ell'']$ and $\widehat{W}\!\in\!\mathrm{L}[j\!-\!1,\!i'-i'',\!\ell'-\ell'']$.
\item $w_j$ is not a leaf in $C$, $\ell''\!\geq\!1$, $\widehat{U}\!\in\!\mathrm{M}[w_j,\!w_j,\!i'',\!\ell'']$ and $\widehat{W}\!\in\!\mathrm{L}[j,\!i'\!-\!i'',\!\ell'\!-\!\ell'']$.
\end{enumerate}
We get that $\widehat{X}\!=\!\widehat{U}\!\cup\!\widehat{W}\!\in\!{\cal A}$ and $\widehat{X}\!\cap\! Y\!=\!\emptyset$; thus, Claim \ref{claim:kl2} is correct.\qed

\mysection{Improving Known Applications}\label{app:knownapp}

We now show that by simply replacing the computation of representative families of Fomin et al. \cite{representative} with our faster computation, we get improved algorithms for {\sc Long Directed Cycle}, {\sc Weighted $k$-Path} and {\sc Weighted $k$-Tree}.

\mysubsection{Long Directed Cycle}

Given a graph $G\!=\!(V,\!E)$ and a parameter $k\!\in\!\mathbb{N}$, we need to decide if $G$ contains a simple cycle of length at least $k$.

By using the algorithm for {\sc Long Directed Cycle} given in \cite{representative} with our computation scheme, \alg{RepAlg}, we immediately get that {\sc Long Directed Cycle} can be solved in time

\[O(2^{o(k)}|E|\log^2 |V|\cdot\max_{0\leq t\leq k}\left\{\frac{(2ck)^{2k}}{t^t(2ck-t)^{2k-t}}(\frac{2ck}{2ck-t})^{2k-t}\right\}).\]
We choose $c=1.5$. Then, the maximum is obtained at $t=k$. Thus, we solve {\sc Long Directed Cycle} in time $O(6.75^k|E|\log^2|V|)$ (improving the previous $O(8^k|E|\log^2|V|)$ time).

\mysubsection{Weighted $k$-Path}

Given a graph $G\!=\!(V,\!E)$, a function $w\!: E\!\rightarrow\!\mathbb{R}$ and a parameter $k\!\in\!\mathbb{N}$, we need to find the minimal weight of a simple path of length $k$ in $G$.

By using the algorithm for {\sc Weighted $k$-Path} given in \cite{representative} with our computation scheme, \alg{RepAlg}, we immediately get that {\sc Weighted $k$-Path} can be solved in time

\[O(2^{o(k)}|V|\log^2 |V|\cdot\max_{0\leq t\leq k}\left\{\frac{(ck)^{k}}{t^t(ck-t)^{k-t}}(\frac{ck}{ck-t})^{k-t}\right\}).\]
We choose $c=1.447$. Then, the maximum is obtained at $t=\alpha k$, where $\alpha\cong0.55277$. Thus, we solve {\sc Weighted $k$-Path} in time $O(2.61804^k|V|\log^2\!|V|)$ (improving the previous $O(2.85043^k|V|\!\log^2\!|V|)$ time).

\mysubsection{Weighted $k$-Tree}

Given a graph $G\!=\!(V,\!E)$, a function $w\!: E\!\rightarrow\!\mathbb{R}$ and a tree $T$ on $k$ nodes, we need to find the minimal weight of a subtree of $G$ isomorphic to $T$.

By using the algorithm for {\sc Weighted $k$-Tree} given in \cite{representative} with our computation scheme, \alg{RepAlg}, we immediately get that for any $\epsilon > 0$, {\sc Weighted $k$-Tree} can be solved in time

\[O(2^{\epsilon k}(\frac{1+\epsilon}{1-\epsilon})^k|V|^{O(\frac{1}{\epsilon})}\cdot\max_{0\leq t\leq k}\left\{\frac{(ck)^{k}}{t^t(ck-t)^{k-t}}(\frac{ck}{ck-t})^{k-t}\right\}).\]
We choose $c=1.447$. Then, the maximum is obtained at $t=\alpha k$, where $\alpha\cong0.55277$. By choosing a small enough $\epsilon>0$, we can solve {\sc Weighted $k$-Tree} in time $O(2.61804^k|V|^{O(1)})$ (improving the previous $O(2.85043^k|V|^{O(1)})$ time).

\end{document}